\documentclass[11pt, oneside, english]{article}

\usepackage[a4paper]{geometry}
\usepackage[english]{babel}

\usepackage{graphicx}
\graphicspath{{./img/} {./img/plot}}

\usepackage{parskip}

\usepackage{bm}
\usepackage{pdfpages}
\usepackage{siunitx}
\usepackage{commath}
\usepackage{amsmath}
\usepackage{amsthm}
\usepackage{amssymb}
\usepackage{booktabs}
\usepackage{mathtools}  

\newcounter{dummy}
\theoremstyle{plain}
\newtheorem{proposition}[dummy]{Proposition}
\newtheorem{lemma}[dummy]{Lemma}
\newtheorem{theorem}[dummy]{Theorem}
\newtheorem{corollary}[dummy]{Corollary}
\newtheorem*{example}{EXAMPLE}

\theoremstyle{definition}

\newtheorem{remark}[dummy]{Remark}
\newtheorem{note}[dummy]{Note}

\usepackage{xpatch} 
\makeatletter
\xpatchcmd{\TPT@doparanotes}{%
\hskip 1em\@plus .3em}{%
\hskip 0.75em\@plus .25em}
\makeatother

\usepackage{float}
\usepackage{subcaption}

\usepackage{setspace}

\newlength{\shiftwidth}
\usepackage{hyperref}
\hypersetup{colorlinks=true, linkcolor=black}

\DeclareMathOperator{\dom}{dom}
\newcommand{\dd}{\dif}

\DeclareMathOperator{\sgn}{sgn}

\DeclareMathOperator{\pp}{\ensuremath{\partial}}
\DeclareMathOperator{\grad}{grad}
\DeclareMathOperator{\rot}{rot}
\DeclareMathOperator{\ddiv}{div}

\DeclareMathOperator{\arctanh}{arctanh}
\newcommand{\sess}{\sigma_{\mathrm{ess}}}

\DeclareMathOperator{\e}{e}

\newcommand{\iu}{\mathrm{i}\mkern1mu}
\newcommand{\vphi}{{\varphi}}

\newcommand{\ob}{\overline}

\newcommand\restr[2]{{
\left.\kern-\nulldelimiterspace 
#1 
\right|_{#2} 
}}

\usepackage{caption}
\usepackage{authblk}
\author[1]{Tomáš Faikl}
\affil[1]{Department of Mathematics, Faculty of Nuclear Sciences and Physical Engineering,\protect\\ Czech Technical University in Prague}
\affil[ ]{\textit {tomas.faikl@fjfi.cvut.cz}}
\title{Spectral analysis of metamaterials in curved manifolds}

\hypersetup{%
    linkcolor=red,          
    citecolor=blue,        
    urlcolor=cyan        
}

\begin{document}
\maketitle
\abstract{
Negative-index metamaterials possess a negative refractive index and thus
present an interesting substance for designing uncommon optical effects such as
invisibility cloaking. This paper deals with operators encountered in an
operator-theoretic description of metamaterials. First, we introduce an
indefinite Laplacian and consider it on a compact tubular neighbourhood in
constantly curved compact two-dimensional Riemannian ambient manifolds, with
Euclidean rectangle in $\mathbb{R}^2$ being present as a special case. As this
operator is not semi-bounded, standard form-theoretic methods cannot be
applied. We show that this operator is (essentially) self-adjoint via
separation of variables and find its spectral characteristics. We also provide
a new method for obtaining alternative definition of the self-adjoint operator
in non-critical case via a generalized form representation theorem. The main
motivation is existence of essential spectrum in bounded domains.}


\section{Introduction}
Metamaterials are artificially created materials that exhibit properties not found in natural substances. These exotic properties are usually limited to a certain electro-magnetic frequency range and they typically arise as a result metamaterial's internal structure, which typically features elements smaller than the respective wavelength. The concept of metamaterials was first theoretically explored by Viktor Veselago in 1967~\cite{veselago1967electrodynamics}, who demonstrated that materials with negative electric permittivity $\epsilon<0$ and magnetic permeability $\mu<0$ exhibit a negative refractive index. Maxwell's equations adapted for metamaterials reveal that the Poynting vector, which indicates the direction of energy flow, points opposite to the wave propagation direction. This counter-intuitive behavior is responsible for the remarkable properties of metamaterials. Interesting applications of the phenomena include effects such as superlensing~\cite{pendry2000negative} for super-resolution microscopy, metamaterial object cloaking~\cite{milton2006cloaking} --- effectively hiding the object from outer observers and reversed Doppler effect or reversed Cherenkov radiation.  Such materials have appeared in physical experiments since 1999, when J. Pendry and his team~\cite{pendry1999magnetism} provided practical designs for metamaterials with negative permittivity and permeability. 

This paper aims to provide a rigorous mathematical framework for understanding the spectral theory of the interface between conventional materials and metamaterials on two-dimensional surfaces in case of rectangular domains. The specific choice of domain is motivated by difficulties of handling sharp corners in existing literature on the subject and by effect of the surface curvature on the results. By focusing on operators involved in description of the interface between these materials, we seek to elucidate the underlying physical principles via spectral theory in various geometric configurations.

We will entertain the quasi-static approximation to Maxwell equations for a scalar electric potential. In this framework, the electric and magnetic fields are no longer dependent on the counterpart's field time derivatives. This way, the problems for electric and magnetic field separate. In the following, we will work only with the electric field
\begin{equation}
    \ddiv \vec D =\rho, \qquad \rot \vec E=0.
\end{equation}
These equations are the Gauss and Faraday law without presence of magnetic field. From Poincaré lemma, we see that the electric field $\vec E=-\grad \phi$ can be described by an electric potential $\phi$. Combining it with relation for homogeneous material $\vec D=\epsilon \vec E$, we obtain equation for the potential $\phi$:
\begin{equation}\label{eq:diveqn}
    -\ddiv(\epsilon\grad  \phi)=\rho.
\end{equation}

Such expressions make appearance in various areas of physics. The main motivation of this work is the metamaterial cloaking phenomenon in electromagnetism, for mathematical-oriented survey of recent progress, see~\cite{nguyen2017negative} and~\cite{felbacq2017metamaterials} for general information on metamaterials. Regarding connections to metamaterial cloaking, we refer the reader to a nice mathematical summary in~\cite{cacciapuoti2019self, nguyen2019cloaking, li2015quasi, ammari2013spectral, milton2006cloaking, milton2002theory}. Our model on curved two-dimensional surfaces can be viewed as a result of electromagnetism theory in curved two-dimensional background.
Also, the model can be viewed as a Hamiltonian of quantum particle in nanostructures with effectively negative mass, as in~\cite{znojil2012schrodinger}, optionally constraining the particle to a curved surface.

Regarding the mathematical justification of negative permitivity and/or permeability appearing in Maxwell equations and in this model, these parameters are negative only effectively --- they are negative in a sense of homogenisation, i.e. only when electromagnetic waves have wavelength much greater than typical distance of the metamaterial structure. See references~\cite{bouchitte2009homogenization,bouchitte2010homogenization,felbacq2005theory,lamacz2013effective,lamacz2016negative} for split-ring resonators and bulk dielectric inclusions to achieve effectively negative parameters also near resonant effects in the media in both $\Omega\subset\mathbb{R}^2$ and $\Omega\subset\mathbb{R}^3$.

Our model takes place on a general surface realized as a two-dimensional Riemannian manifold $\mathcal{M}$. We then define a bounded rectangle, as tubular $(a,b)$-neighbourhood of curve $\Gamma$, on this surface denoted by $\Omega$. It can be illustrated as a non-symmetric rectangle and we refer the reader to Figure~\ref{img:manifold-fermi}. The coordinates $(x_1,x_2)\equiv(x,y)$ are called Fermi or geodesic parallel coordinates and details are provided in Section~\ref{section:setting}. 
For $a$, $b$, $c>0$, define the domains
\begin{equation}\label{eq:def-rectangle}
    \Omega_0:=(-b,a)\times(0,c)\equiv J_1\times J_2, \quad \Omega_+ := (0,a)\times J_2, \quad \Omega_- := (-b,0)\times J_2
\end{equation}
in $\mathbb{R}^2$ and we let $\mathcal{C}:=\{0\}\times J_2$ be the interface between $\Omega_\pm$. 
\begin{figure}
    \centering
    \includegraphics[width=1\linewidth]{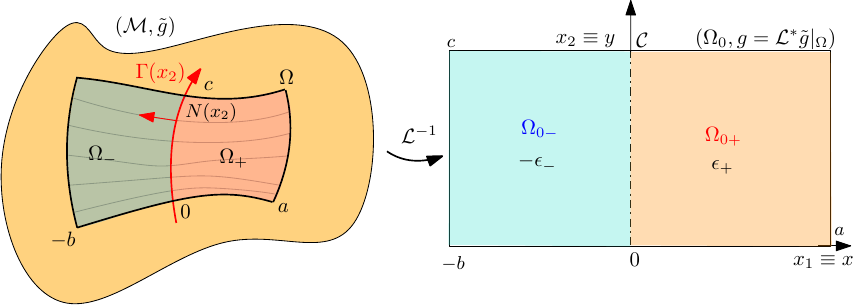}
    \caption{Curve $\Gamma$ on a Riemannian manifold $(\mathcal{M}, \tilde g)$ and its tubular neighbourhood $\Omega$. At every point of $\Gamma$, there exists a geodesic perpendicular to $\Gamma$ which is used to construct a rectangle on the manifold. The tubular neighbourhood $\Omega$ is diffeomorphic to a rectangle $\Omega_0$ in Fermi coordinates $x_1,x_2$ (on the right) with induced diagonal metric $g$. Details are provided in Section~\ref{section:setting}.}
    \label{img:manifold-fermi}
\end{figure}
Now, define an indefinite Laplacian on this \emph{rectangle}. The action of the indefinite Laplacian is formally given by 
\begin{align}\label{eq:divgrad-epsilon}
    \begin{split}
        A&=-\ddiv_g \left(\epsilon \nabla \right),\\
        \epsilon(x)&=
        \begin{cases}
            \epsilon_+ \quad &x \in \Omega_+,\\
            - \epsilon_- \quad &x \in \Omega_-,\\
        \end{cases}
        \quad\epsilon_\pm>0.
    \end{split}
\end{align}
with piecewise constant function $\epsilon$ representing jump in permitivity of material-metamaterial transition and $\ddiv_g$ being divergence associated to metric $g$. Let us naturally identify $L^2$ spaces on corresponding subdomains with measure induced by the metric, $L^2(\Omega_0,g)\simeq L^2(\Omega_+,g) \oplus L^2(\Omega_-,g)$. The differential expression~\eqref{eq:divgrad-epsilon} can be alternatively given for sufficiently differentiable function $\psi\in L^2(\Omega_0,g)$ in terms of the Laplace-Beltrami operator $\Delta_g$ associated to metric $g$.

For the moment, we restrict ourselves to the special case of constantly-curved manifolds -- even though we provide generalizations of some key results also to richer geometric domains and non-constant curvatures later. Fix a constant Gaussian curvature $K\in\mathbb{R}$ and let the metric $g$ correspond to $K$. Define an intermediate operator $\dot A_K: \dom \dot A_K \subset L^2(\Omega_0,g)\to L^2(\Omega_0,g)$ by
\begin{equation}
\label{eq:def-dotA_K}
    \begin{aligned}
        &{\hskip13em\relax}\dot A_K \begin{pmatrix}\psi_+\\\psi_-\end{pmatrix}=\begin{pmatrix}-\epsilon_+ \Delta_g \psi_+\\\epsilon_- \Delta_g \psi_-\end{pmatrix},\\
        &\dom \dot A_K:= \left\{
        \begin{array}{l|l}
         &\psi_\pm\mid_{\pp \Omega_0}=0, \\
        \psi=\begin{pmatrix}\psi_+\\\psi_-\end{pmatrix} \in H^2(\Omega_+,g)\oplus H^2(\Omega_-,g) &\psi_+ = \psi_- \ \mathrm{on\ } \mathcal{C}\\
        &\epsilon_+ \pp_x\psi_+ = -\epsilon_- \pp_x\psi_- \ \mathrm{on\ } \mathcal{C}
        \end{array}
    \right\}
    \end{aligned}
\end{equation}
where the main peculiarity is the sign-changing derivative on the interface. These precise interface conditions are needed in order for the differential expression to be well-defined. As a byproduct, it also corresponds to interface conditions for electric potential in Maxwell theory.

This paper concerns itself with finding self-adjoint Dirichlet realisation $A_K$ of the intermediate operator $\dot A_K$ along with its essential spectrum. Mathematically, the operator considered is symmetric, although it does not possess ellipticity, nor is semi-bounded and hence, standard form-theoretic methods do not apply directly\footnote{Although some representation theory of indefinite quadratic forms is developed in~\cite{grubivsic2013representation}.}. For further discussion, it is beneficial to define \emph{contrast} $\kappa>0$:
\begin{equation}\label{eq:def-contrast}
    \kappa := \frac{\epsilon_+}{\epsilon_-},
\end{equation}
as the results depend crucially on the \emph{criticality}, or \emph{non-criticality} of the contrast. For the rectangle~\eqref{eq:def-rectangle} it is the value of $\kappa=1$ which is \emph{critical} in some sense defined below and this case is usually much harder to reason about.

The mathematical problem of properly defining such operators has an enduring history along with some recent activity in the subject. 
It is known from work by Behrndt and Krejčiřík~\cite{behrndt2014indefinite} that the indefinite Laplacian on a rectangle in flat underlying space $\mathbb{R}^2$, in our notation $A_0$, with $a=b$, has zero as an infinitely-degenerate eigenvalue exactly when the contrast $\kappa$ is \emph{critical}, i.e. $\kappa=1$. Otherwise, the essential spectrum is empty. That is an unusual effect on bounded domain caused by a domain transmission condition on the interface. In their paper, they find its self-adjoint extension by employing separation of variables and Krein-von Neumann extensions. It was found that the functions in domain of the self-adjoint realisation of $\dot A_0$ do not belong to any local Sobolev space $H^s$, $s>0$.

One of the first conducted mathematical research of properties of such indefinite Laplacians for \emph{non-critical} contrast was published in 1999~\cite{bonnet1999analyse}. There, the authors consider the problem in $\mathbb{R}^2$, although in a \emph{different domain}:
\begin{align}
  \begin{split}
    \dom A&=\left\{u\in H_0^1(\Omega): \mathrm{div} \left(\epsilon \nabla u\right)\in L^2(\Omega)\right\},\\
    Au=-\mathrm{div} \left(\epsilon\nabla u\right)&,\quad \forall u\in\dom A,\qquad \epsilon(x):=
    \begin{cases}
      \epsilon_+,&x\in\Omega_+,\\
      -\epsilon_-,&x\in\Omega_-,
    \end{cases}
  \end{split}
\end{align}
for $\epsilon_\pm>0$, $\kappa\not=1$ such that $\Omega=\Omega_+\cup\,\Omega_-$, boundary $\Sigma$ of $\Omega$ sufficiently regular, boundary $\Gamma_-$ of $\Omega_-$ Lipschitz continuous and $\Sigma\cap\Gamma_-=\emptyset$. For this non-critical contrast, the resulting operator is self-adjoint, has a compact resolvent and its eigenvalues are accumulating to $\pm\infty$. It is also of interest that when interface $\Gamma$ is not smooth, for example when there is a right-angle \emph{corner} on $\Gamma$, then the results extend to values of contrast $\kappa:=\frac{\epsilon_+}{\epsilon_-}$ which do not belong to some interval containing the critical contrast of 1. For values of contrast $\kappa$ inside the critical interval, the operator $A$ is not self-adjoint.

Complementary results for \emph{smooth} $\Omega$ and $\Omega_\pm\in\mathbb{R}^n$ with \emph{smooth} interface $\Gamma$ is discussed and solved in~\cite{cacciapuoti2019self} by~Cacciapuoti, Pankrashkin and Posilicano via method of boundary triplets. In two dimensions, 0 is in the essential spectrum whenever $\kappa=1$. In higher dimensions, it is an intricate effect of the geometry of the interface $\Gamma$. Although the 1-dimensional case is fundamentally different, another paper~\cite{hussein2014sign} deals with quantum graph networks and as a special case, it contains the situation $\Omega\subset \mathbb{R}$. It is found that the operator has empty essential spectrum regardless of contrast $\kappa$.

Now, we bring forth the main results of the paper. None of the results depend on the curvature in a crucial way.
One technique of defining the operator is by separation of variables on the rectangle and decomposing $\dot A_K$ into an (infinite) sum of transversal one-dimensional operators $A_K^n$ with compact resolvent, $n\in\mathbb{N}$, acting in $L^2(J_1)$. This technique allows us to define the self-adjoint extension for arbitrary contrast $\kappa>0$.
\begin{theorem}\label{theorem:intro-essential}
    Let $K\in\mathbb{R}$. Operator $\dot A_{K}$ given in~\eqref{eq:def-dotA_K} is essentially self-adjoint. Denote $A_K:=\overline{\dot A_K}$ its closure. Then $A_K$ satisfies the following:
    \begin{enumerate}
        \item its eigenfunctions form a complete orthonormal set in $L^2(\Omega_0, g)$,
        \item $\sigma(A_K)=\cup_{n\in\mathbb{N}}\ \sigma\left(A_K^n\right)$.
        \item operator $A_K$ can be expressed as a direct sum of transversal one-dimensional self-adjoint operators with compact resolvent $A_K^n$, $n\in\mathbb{N}$, using formula
            \begin{equation}
                A_K = \bigoplus_{n=1}^\infty \left(\mathbb{I}\otimes A_K^n\right) P_n,
            \end{equation}
            where $P_n: L^2(\Omega_0,g)\to L^2(\Omega_0,g)$ are rank-one operators; for $\Psi\in L^2(\Omega_0,g)$
            \begin{equation}
                (P_n\Psi)(x,y):=\left(\phi_m, \Psi(x,\cdot)\right)_{L^2(J_2)}\phi_n(y),\quad \phi_n(y):=\sqrt{\frac{2}{c}}\sin\left(\frac{n\pi}{c}y\right).
            \end{equation}
    \end{enumerate}
\end{theorem}
It is of interest that Theorem~\ref{theorem:intro-essential} can be generalized to situation where the Gaussian curvature is not constant as long as some mild regularity condition on $K(x)$ is satisfied.

The following Theorem quantifies an uncommon effect for differential operators defined on a bounded domain.
\begin{theorem}
    Let $K\in\mathbb{R}$. Then $\kappa=1$ implies $0\in\sigma_\mathrm{ess}(A_K)$. In manifold with zero Gaussian curvature $K=0$, we have precisely $\sigma_\mathrm{ess}(A_{K=0})=\{0\}$ and there is an exponential decay of the smallest (in absolute value) eigenvalue of $A_K^n$, $\lambda_n=o\left(\exp\left(-\frac{n\pi}{c}\min\{a,b\}\right)\right)$.
\end{theorem}
The detailed analysis below implies that for critical contrast and zero curvature $K=0$, $\lambda=0$ is an eigenvalue precisely when $a=b$ (then it is an infinitely-degenerate eigenvalue).

    \begin{example}
        For critical contrast, $a=b$ and $K=0$, we provide eigenfunctions from the eigensubspace corresponding to the infinitely-degenerate eigenvalue $\lambda=0$. The eigenfunctions are precisely those given by $f_n(x,y)=\mathcal{N}_n\cdot\phi_n(y)\psi_n(x)$ where $\phi_n(y)=\sqrt \frac{2}{c}\sin(\frac{n\pi}{c}y)$,
            \begin{equation}
                \psi_n(x)=
                \begin{cases}
                    \sinh\left(\frac{n\pi}{c}(a-x)\right), & x\ge0 \\
                    \sinh\left(\frac{n\pi}{c}(a+x)\right), & x<0
                \end{cases}
                %
            \end{equation}
            and $\mathcal{N}_n$ are norm constants so that $||f_n||=1$ for all $n\in\mathbb{N}$. These functions are localized near the interface for large $n\in\mathbb{N}$.
    \end{example}
    For critical contrast and $a\not=b$, 0 is not an eigenvalue, although it is an accumulation point of the spectrum.

As a next step, we provide --- only for non-critical contrast --- a new technique for defining the self-adjoint extension via forms. In the end, we will show that the two resulting operators (the other one due to separation of variables) coincide for non-critical contrast. The technique used is similar to $\mathbb{T}$-coercivity used in~\cite{Bonnet_Ben_Dhia_2010,dhia2012t} to provide a well-posedness for a similar problem and more general domains in $\mathbb{R}^2$ and $\mathbb{R}^3$. Combining $\mathbb{T}$-coercivity and arguments involving smooth partitions of unity, the authors there derive criteria for well-posedness of a similar problem in terms of quantities ($\epsilon$) in the neighbourhood of the boundary.

We apply a similar, although new approach using a generalized Lax-Milgram theorem of Almog and Helffer~\cite{almog2015spectrum} with a mirroring and a cut-off technique. Note that $\mathbb{T}$-coercivity is a special case of the referenced representation theorem.
\begin{theorem}\label{theorem:intro-forms}
    For contrast $\kappa\not = 1$, there is a unique self-adjoint operator $\mathcal{A}_K: \dom  \mathcal{A}_K \subset L^2(\Omega_0,g)\to L^2(\Omega_0,g)$, $\dot A_K\subset \mathcal{A}_K$, associated to the sesquilinear form $a$ given by
    \begin{equation}
        a(u,v)=:(u,\mathcal{A}_K v)_{L^2(\Omega_0,g)}, \quad u\in H_0^1(\Omega_0,g),\ v\in\dom \mathcal{A}_K\subset H_0^1(\Omega_0,g)
    \end{equation}
    with domain 
    \begin{equation}
        \dom  \mathcal{A}_K:= \left\{v\in H_0^1(\Omega_0,g): \Delta v_\pm \in L^2(\Omega_\pm,g),\ \restr{\left(\epsilon_+\partial_x v_+ +\epsilon_-\partial_x v_-\right)}{\mathcal{C}} =0\right\}
    \end{equation}
    where the interface condition is understood in the weak sense of 
    \begin{align}
        \begin{split}
            \Delta v_\pm\in &L^2(\Omega_\pm,g),\ \restr{\left(\epsilon_+\partial_x v_+ +\epsilon_-\partial_x v_-\right)}{\mathcal{C}} =0 \\
                &\vcentcolon\iff \forall u\in H_0^1(\Omega_0,g):\ \int_{\Omega_0} \overline{\nabla u}\epsilon\nabla v = - \int_{\Omega_0} \overline{u} \nabla\cdot\left(\epsilon\nabla v\right)
        \end{split}
    \end{align}
    and $\mathcal{A}_K$ has compact resolvent and $0\not\in\sigma\left(\mathcal{A}_K\right)$. Furthermore, for $\kappa\not=1$, we have $\mathcal{A}_K=A_K$ where $A_K$ is defined in Theorem~\ref{theorem:intro-essential}. 
\end{theorem}
Theorem~\ref{theorem:intro-forms} can be generalized to non-rectangular geometries and non-constant curvatures, an example of which is in Figure~\ref{fig:form-generalization}.
\begin{figure}
    \centering
    \includegraphics[width=0.5\textwidth]{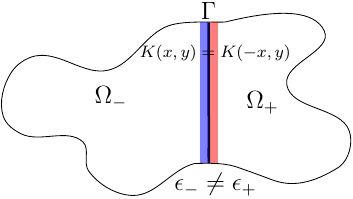}
    \caption{Form approach in non-critical case can be extended to non-constantly curved ambient manifold under assumption that there is a rectangular neighbourhood of the curve $\Gamma$ contained in $\Omega$ such that metric satisfies $g(x,y)=g(-x,y)$ or equivalently for curvature, $K(x,y)=K(-x,y)$ in the neighbourhood in Fermi coordinates.}
    \label{fig:form-generalization}
\end{figure}

\vspace{10pt}

The paper is organized as follows. In Section~\ref{section:setting}, we provide formal definitions of the underlying geometric and functional spaces and some basic results to gain a bit more insight into the problem. We proceed to Section~\ref{section:essential-self-adjoint} to provide means of defining the self-adjoint extension $A_K$ via separation of variables. Form representation for non-critical contrast is contained in Section~\ref{section:form-approach} and essential spectrum results are given in Section~\ref{section:critical-contrast}.

\subsection{Related work}
Beside the references mentioned in the introduction, in~\cite{dhia2007two} the authors explore well-posedness of system (for various bounded domains $\Omega$)
\begin{equation}
  \mathrm{div}\left(\frac{1}{\epsilon}\nabla u\right)+\omega^2\mu u=f \quad \mathrm{in}\ \Omega
\end{equation}
 with Dirichlet boundary conditions in $H^1_0(\Omega)$ for the case of sign-changing permitivity $\epsilon$ on the interface and source $f\in L^2(\Omega)$. The problem was reformulated in a variational approach and allowed to tackle non-constant permitivities $\epsilon_\pm$ and Lipschitz-regular interface $\Gamma$. In a following paper~\cite{Bonnet_Ben_Dhia_2010, dhia2012t}, the authors applied the framework of $\mathbb{T}$-coercivity and interface-localisation techniques. A large quantity of examples was provided for domains in both $\mathbb{R}^2$ and $\mathbb{R}^3$. Finally, similar results were derived for the full non-scalar Maxwell problem in~\cite{dhia2014t} in a time-harmonic case. Important observations are that the time-harmonic Maxwell problem can be fully solved in terms of scalar problems which highlights importance of studying the scalar problems introduced so far.

 \paragraph{Acknowledgements.} The author is grateful to David Krejčiřík, supervisor of his bachelor and diploma theses, for providing fruitful discussions and motivation. This article is the outcome and could not be finished without his help and support. Author is also grateful to Petr Siegl for ideas in the form approach. The author was supported by the \emph{EXPRO Grant No. 20-17749X} of the Czech Science Foundation.

\numberwithin{equation}{section}
\numberwithin{dummy}{section}

\section{Geometrical and functional problem setting}\label{section:setting}
\subsection{Geometrical setting}
We will use similar geometrical setting as in~\cite{krejcirik2010ptsymmetric}. In this section, we will make frequent use of Sobolev spaces on manifolds referenced in~\cite{hebey2000nonlinear} and Fermi (geodesic parallel) coordinates~\cite{hartman1964geodesic} in tubular neighbourhoods~\cite{gray2003tubes}.

We have defined the rectangular domain $\Omega_0\subset\mathbb{R}^2$ and other notions in~\eqref{eq:def-rectangle}. Overall, we have a disjoint union $\Omega_0=\Omega_{-}\cup\mathcal{C}\cup\Omega_{+}$. The metamaterial is located in $\Omega_{-}=(-b,0)\times J_2$ and material with positive permitivity is located in $\Omega_{+}=(0,a)\times J_2$. 

Consider a two-dimensional Riemannian manifold $\mathcal{M}$ and assume that its Gaussian curvature $K$ is continuous (which holds if $\mathcal{M}$ is $C^3$-smooth or is embedded into $\mathbb{R}^3$). Additionally, let $\Gamma: J_2\to\mathcal{M}$ be a $C^2$ curve parametrized by arc length. This curve $\Gamma$ will serve as the metamaterial interface.
Let us introduce a tubular neighbourhood $\Omega$ of curve $\Gamma$. In case of $a=b$, $\Omega$ can be seen as a set of points on $\mathcal{M}$ with geodesic distance less than $a$ from $\Gamma$. Construct a mapping $\mathcal{L}: \Omega_0\to\mathcal{M}$ given by
\begin{equation}\label{eq:manifold-expmap}
    \mathcal{L}(x_1,x_2) := \exp_{\Gamma(x_2)}\left(x_1N(x_2)\right),\quad (x_1,x_2)\in\Omega_0
\end{equation}
where $\exp_q$ is the exponential map of $\mathcal{M}$ at point~$q\in\mathcal{M}$ and $N(x_2)\in \textrm{T}_{\Gamma(x_2)}\mathcal{M}$ is a normal vector to curve $\Gamma$ in $x_2\in J_2$, an element of tangent space to manifold  $\mathcal{M}$. The coordinates are chosen so that $\mathcal{L}(\mathcal{C})=\Gamma$. Then 
\begin{equation}
    \Omega:=\mathcal{L}(\Omega_0)
    \label{eq:omega-on-manifold}
\end{equation}
and $(\Omega_0,g)$ is a Riemannian manifold with induced metric $g:=\mathcal{L}^*\restr{\tilde g}{\Omega}$ from $\mathcal{M}$. The coordinates $(x_1,x_2)\equiv(x,y)$ are called Fermi or geodesic parallel coordinates.

In the following text, $\mathcal{L}:\Omega_0\to\Omega$ will be assumed to be a diffeomorphism, although the condition can be weakened as seen below. The map $\mathcal{L}$ is always a diffeomorphism, provided $a$, $b$ are small enough. Set $\Omega$ can be parametrized via the geodesic parallel coordinates $(x_1,x_2)$. From Gauss lemma, it follows that
\begin{align}\label{eq:metric-g}
  \begin{split}
    g=
    \begin{pmatrix}
      1&0\\
      0&f^2
    \end{pmatrix}
  \end{split}
\end{align}
where $f$ is continuous and has continuous partial derivatives $\partial_1 f$,  $\partial_1 ^2 f$ satisfying the Jacobi equation
\begin{equation}
    \pp_1^2f + Kf=0 \quad \land \quad 
    \begin{cases}
        &f(0, \cdot)=1 \\
        &\pp_1f(0, \cdot)=-\kappa,
    \end{cases}
    \label{eq:Jacobi-curvature}
\end{equation}
where $K$ is Gaussian curvature at a point with local coordinates $(x_1,x_2)$ and $\kappa$ is the geodesic curvature of $\Gamma$. The solutions for constant Gaussian curvatures are
\begin{equation}
    f(x_1, x_2) =
    \begin{cases}
        \cos(\sqrt{K}x_1) - \frac{\kappa(x_2)}{\sqrt{K}} \sin(\sqrt{K}x_1) & \mathrm{ if }\ K~> 0, \\ 
        1- \kappa(x_2)\cdot x_1 & \mathrm{ if }\ K~= 0, \\
        \cosh(\sqrt{|K|}x_1) - \frac{\kappa(x_2)}{\sqrt{|K|}} \sinh(\sqrt{|K|}x_1) & \mathrm{ if }\ K~< 0,
    \end{cases}
\end{equation}
and from now on, we will assume that the geodesic curvature of $\Gamma$ is identically zero, i.e.
\begin{equation}
  \boxed{\mathrm{curve}\ \Gamma\ \mathrm{is\ a\ geodesic}.}
\end{equation}
A manifold $(\mathcal{M},g)$ with arbitrary $K\in\mathbb{R}$ is diffeomorphic to one with $K\in\left\{-1,0,1\right\}$. Hence, up to diffeomorphism, we can setup our problem in a $L^2(\Omega_0, g)$ space with measure $\dd \nu_g$ given by
\begin{equation}
    \dd\nu_g:=
    \begin{cases}
        \cos(x_1)\dd x_1 \dd x_2, & \mathrm{if} \ K=+1, \\
        \dd x_1 \dd x_2, &\mathrm{if} \ K=0, \\
        \cosh(x_1) \dd x_1 \dd x_2, & \mathrm{if} \ K=-1,
    \end{cases}
\end{equation}
and in order for $g$ to be a positive-definite metric, we have to restrict the dimensions of the rectangle in case $K=+1$ to 
\begin{equation}
  \boxed{K=+1\implies a,b\in\left(0,\frac{\pi}{2}\right).}
\end{equation}

The assumption of $\mathcal{L}$ being a diffemorphism also means that geodesics on $\Omega$ do not intersect as then $\mathcal{L}$ would not be bijective. This assumption can be weakened by only requiring $f,f^{-1}\in L^\infty(\Omega_0,g)$. 


As the domain $\Omega$ can be covered by a single choice of (Fermi) global coordinates, we will sometimes omit the index of $\Omega_0$ and write only $\Omega$ and we will often use notation for the coordinates on $\Omega_0$ as $x_1\equiv x$,  $x_2\equiv y$.

\begin{figure}
  \centering
  \begin{subfigure}{0.32\textwidth}
      \captionsetup{skip=25pt}
    \centering
    \includegraphics[width=\textwidth,trim=1.5cm 3cm 1cm 0, clip]{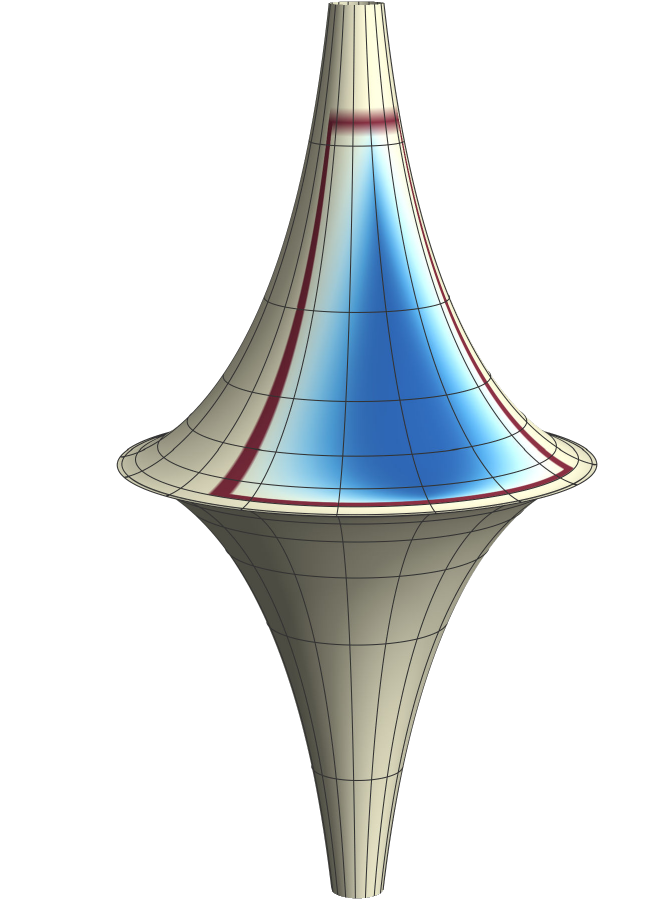}
    \caption{A \emph{pseudosphere}, $K=-1$.}
    \label{fig:subfig_a}
  \end{subfigure}
  \hfill
  \begin{subfigure}{0.3\textwidth}
      \captionsetup{skip=25pt}
    \centering
    \includegraphics[width=\textwidth,trim=1.5cm 0 .5cm 0, clip]{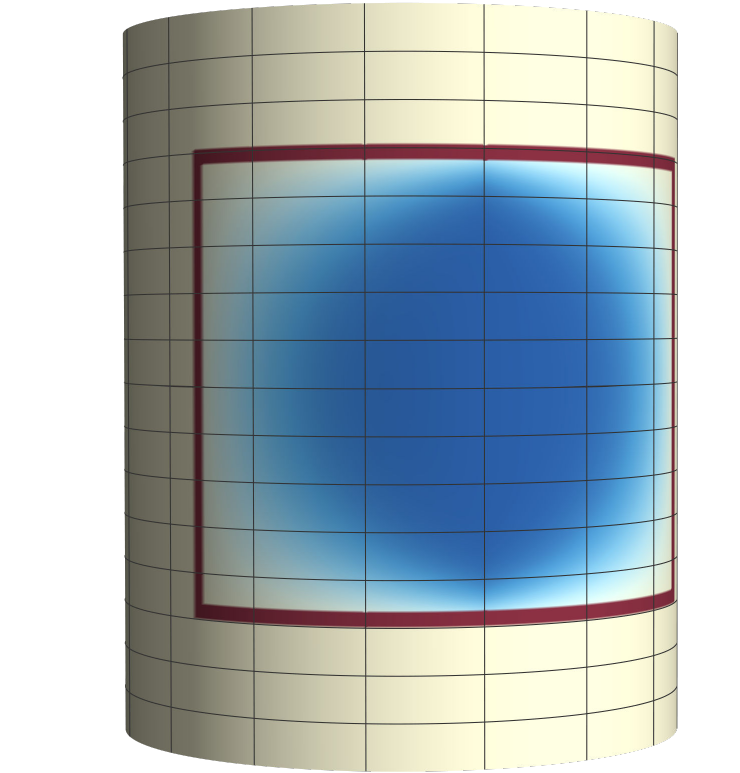}
    \caption{A \emph{cylinder}, $K=0$.}
    \label{fig:subfig_b}
  \end{subfigure}
  \hfill
  \begin{subfigure}{0.35\textwidth}
  \captionsetup{skip=25pt}
    \centering
    \includegraphics[width=\textwidth,trim=0 0 0 0, clip]{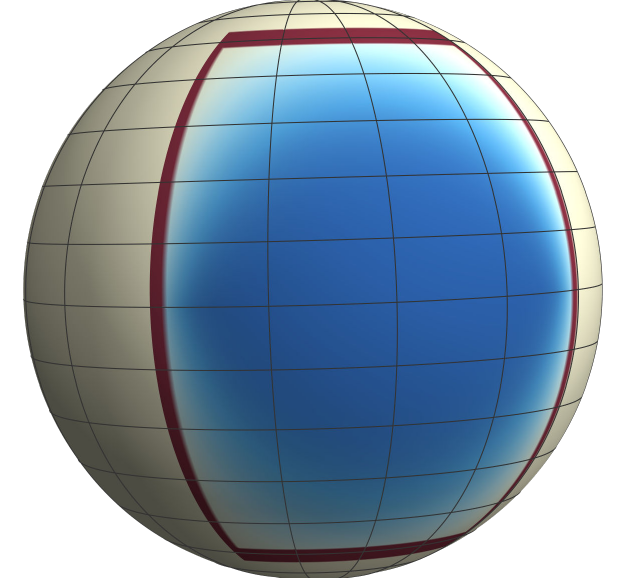}
    \caption{A \emph{sphere}, $K=+1$.}
    \label{fig:subfig_c}
  \end{subfigure}

  \caption{Rectangles as defined by construction~\eqref{eq:manifold-expmap} depicted on various manifolds with constant curvature. The boundary of the rectangle is red and the inside is blue. In fact, the blue color represents values of eigenfunction corresponding to mode $m=1$. }
  \label{fig:full_figure}
\end{figure}

\begin{remark}
    Note that the metric~\eqref{eq:metric-g} has the structure of the metric
    of any two-dimensional Riemannian manifold expressed in
    geodesic polar coordinates. This is consistent because
    our geodesic parallel coordinates are a generalision;
    indeed, the reference point (pole) of the polar coordinates
    is replaced by a submanifold (here a curve).
    The independence of the Jacobian on $x_2$ then corresponds
    to the independence of the Jacobian in polar coordinates
    on the angular variable for rotationally symmetric manifolds.
\end{remark}

\subsection{Hilbert space on manifold \texorpdfstring{$(\Omega_0,g)$}{(Omega0, g)}}
The space in which we will examine the differential operator is that of $L^2(\Omega_0,g)$, the Hilbert space of measurable functions defined on the Riemannian manifold $(\Omega_0.g)$ with the corresponding measure $\dd\nu_g=f\dd x_1\dd x_2$ induced by $g$ such that the norm $\|\cdot\|_g$ induced by the inner product
\begin{equation}
    (u,v)_g:=\int_{\Omega_0}\overline{u(x)}v(x)| \det g|^{\frac{1}{2}}\dd x=\int_{\Omega_0}\overline{u(x)}v(x) f(x)\dd x
\end{equation}
is finite. For $f,f^{-1}\in L^\infty(\Omega_0)$, this norm is equivalent to the standard $L^2(\Omega_0)$ norm. The Sobolev space
\begin{equation}
    W^{1,2}(\Omega_0,g):=\left\{\psi\in L^2(\Omega_0,g) \mid |\nabla\psi|_g^2:=\ob{\pp_i\psi}g^{ij}\pp_j\psi \in L^2(\Omega_0,g) \right\}
\end{equation}
can be identified with the usual Sobolev space $W^{1,2}(\Omega_0)$, frequently also denoted $H^1(\Omega_0)$. Note that, although this requirement is satisfied in our case, that the $W^{2,2}$-Sobolev space on manifold $(\Omega_0,g)$ can be identified with $W^{2,2}(\Omega_0)$ if
\begin{equation}
    \forall x_1\in J_1: \quad f(x_1,\cdot), f^{-1}(x_1,\cdot)\in W^{1,\infty}(J_2).
\end{equation}
For more information and references on the matter, see the original article~\cite{krejcirik2010ptsymmetric}.

\subsection{Indefinite Laplacian on a manifold}
Maxwell equations on the manifold lead to conceptually equivalent differential operator given in terms of $\mathrm{div}$ and $\mathrm{grad}$. In local coordinates on a manifold $(\Omega_0,g)$ we have the following identities for differentiable function $\psi$ and vector field $X$ in local coordinates $(x_1,x_2)$
\begin{align}
    (\grad \psi)^i &= (\dd\psi)^i = g^{ij}\pp_j \psi, \\
    \ddiv X &= \frac{1}{\sqrt{|g|}}\pp_i(\sqrt{|g|}X^i).
\end{align}

Differential expression $-\mathrm{div}(\epsilon \mathrm{grad})$ can be written as
\begin{equation}
    -\ddiv(\epsilon\grad\psi) 
    =
    -\frac{1}{f}\partial_1\left(\epsilon f \frac{\partial\psi}{\partial x^1}\right)
    - \frac{1}{f}\partial_2\left(\epsilon \frac{1}{f} \frac{\partial\psi}{\partial x^2}\right),
    \label{eq:laplace-manifold}
\end{equation}
and for piecewise constant $\epsilon$, it can be given as
  \begin{equation}
      -\ddiv\left(\epsilon\grad \begin{pmatrix}\psi_+\\\psi_-\end{pmatrix}\right) = \begin{pmatrix}-\epsilon_+\Delta_g\psi_+\\\epsilon_-\Delta_g\psi_-\end{pmatrix}
  \end{equation}
  for $\psi\in H^2(\Omega_+,g)\oplus H^2(\Omega_-,g)$ and $\Delta_g$ is a Laplace-Beltrami operator.

  For a Riemannian manifold $(\mathcal{M},g)$ with constant Gaussian curvature, define an initial, restricted operator $\dot A_K: \dom \dot A_K \to L^2(\Omega,g)$ with $\dom \dot A_K$ considered as a subset $\dom \dot A_K\subset L^2(\Omega,g)$ given in~\eqref{eq:def-dotA_K}.
The operator can be written in a unified fashion using expression valid for $K\in\mathbb{R}$ as
\begin{equation}
    \dot A_K=
    \begin{pmatrix}\epsilon_+\\-\epsilon_-\end{pmatrix} \cdot \left(
    - \frac{1}{\cos^2(\sqrt{K}x_1)}\pp_2^2 -\pp_1^2  + \sqrt{K}\tan(\sqrt{K}x_1)\pp_1 \right).
\end{equation}

\subsection{Basic properties of the indefinite Laplacian}\label{section:basic-properties}
The differential operator $\dot A_K$ defined in~\eqref{eq:def-dotA_K} is unbounded and its numerical range is not bounded from below neither from above.
\begin{lemma}
The operator $\dot A_K$ is symmetric for any $K\in\mathbb{R}$, $\kappa>0$.
\end{lemma}
\begin{proof}
    The proof proceeds by an application of the divergence theorem~\cite{lee2012smooth} on $\Omega_{\pm}$ respectively and by using density of $C_0^\infty(\Omega)$ in $\dom\dot A_K$. Ultimately, let $u,v\in\dom\dot A_K$ and
     \begin{align}
        \begin{split}
            (u,\dot A_K v)_{L^2(\Omega_0,g)} &= - \int_{\Omega_0} \overline{u} \ddiv_g \left(\epsilon \nabla v \right) \dd \nu_g \\
                &= - \int_{\Omega_+} \overline{u} \epsilon_+\ddiv_g \left(\nabla v \right) \dd \nu_g + \int_{\Omega_-} \overline{u} \epsilon_-\ddiv_g \left( \nabla v \right) \dd \nu_g\\
                &= \int_{\Omega_0} \epsilon\left(\nabla u, \nabla v\right)_g \dd \nu_g + \int_{\mathcal{C}} \overline{u} \left(\epsilon_+ \pd{v_+}{x} + \epsilon_- \pd{v_-}{x}\right)\dd \tilde \nu_g\\
                &= \int_{\Omega_0} \epsilon\left(\nabla u, \nabla v\right)_g \dd \nu_g = \left(\dot A_K u,v\right)_{L^2(\Omega_0,g)}
        \end{split}
    \end{align}
    where $\left(\cdot,\cdot\right)_g$ is a pairing of 1-forms by metric $g$, $\dd \tilde\nu_g$ is the induced volume form on $\mathcal{C}$ (here $\dd\tilde\nu_g=\dd y$) and the interface integral is zero due to the interface condition in $\dom\dot A_K$. The traces here are well-defined due to $H^2$-regularity on $\Omega_{\pm}$ and $H^1_0$-regularity on $\Omega_0$.
\end{proof}

The following lemma enables us to restrict possible curvatures to $K\in\{-1,0,1\}$.
\begin{lemma}
  Let $(\Omega_0,g)$ be a Riemannian manifold with constant Gaussian curvature $K\not = 0$ and $\dot A$ the operator~\eqref{eq:def-dotA_K}. Then, there exists a homothetic transformation  $\tau: (\Omega_0, g)\to(\tilde\Omega_0,\tilde g)$ of domain $\Omega_0=(-b,a)\times(0,c)$ onto $\tilde\Omega_0=|K|^{\frac{1}{2}}\Omega_0$ so that $g$ corresponds to the original curvature $K$ and $\tilde g$ corresponds to a Gaussian curvature of $\sgn K$ via the Jacobi equation,
    \begin{equation}
    \dot A_K\psi \restr{}{(x_1,x_2)}=\left(|K|\tilde{\dot A}_{(\sgn K)}\tilde\psi\right) \circ \restr{\tau}{(x_1,x_2)},
    \end{equation}
    with tildes denoting object on $(\tilde\Omega_0,\tilde g)$. The eigenvalues of the operator on the original, respectively transformed domain, satisfies
    \begin{equation}
        \lambda\left[\dot A_K(a,b,c) \right]= |K| \cdot \lambda\left[\dot A_{\sgn K}\left(\sqrt{K}a,\sqrt{K}b,\sqrt{K}c\right)\right].
    \end{equation}
\label{lemma:curvature-scaling}
\end{lemma}
\begin{proof}
    Let us assume that $K>0$, the negative case is analogous. The Gaussian curvature of a 2-manifold is invariant to surface reparametrisation. Multiply both sides of the Jacobi equation~\eqref{eq:Jacobi-curvature} by $\frac{1}{K}$ and obtain
    \begin{equation}
        \frac{1}{K}\pp_1^2 f + f = 0.
    \end{equation}
    To find a relation to the operator defined on a manifold with curvature $\sgn K=1$, we need to find a coordinate transformation such that $(x_1,x_2)\mapsto(\tilde x_1, \tilde x_2)$ satisfying $\frac{1}{K}\pp_1^2 f = \tilde\pp_1^2 \tilde f$ (tilde denotes expressions in transformed coordinates). It is straightforward to check that the linear transform $\tau(x_1,x_2)=\left(\sqrt{K} x_1,\sqrt{K} x_2\right)=(\tilde x_1,\tilde x_2)$ satisfies
    \begin{align}
        \begin{split}
            A_K \psi \mid_{(x_1,x_2)}&= \left(-\pp_2^2 - \frac{1}{\cos^2(\sqrt K~x_2)}\pp_1^2 + \sqrt K\tan(\sqrt K~x_2)\pp_2\right) \psi \mid_{(x_1,x_2)} \\
                     &= \left(\left(-K\tilde\pp_2^2-\frac{K}{\cos(\tilde x_2)}\tilde\pp_1^2 + K\tan(\tilde x_2)\tilde\pp_2\right)\tilde\psi\right) \circ \tau \mid_{(x_1,x_2)} \\
                     &= K\left(\tilde{ \dot A}_1 \tilde \psi\right) \circ \tau \mid_{(x_1, x_2)}.
        \end{split}
    \end{align}
\end{proof}
The proof of the following proposition will be completed in the Section~\ref{section:essential-self-adjoint} by showing that the eigenfunctions of $\dot A_K$ form a complete orthonormal set in $L^2(\Omega_0,g)$. This justifies absence of any other eigevalues that those found by using the ansatz in the beginning of the following proof.
\begin{proposition}
    Let $\dot A_0$ be the operator defined in~\eqref{eq:def-dotA_K} for zero Gaussian curvature $K=0$. The operator satisfies the following:
    \begin{enumerate}
        \item Its eigenvalues are
            \begin{equation}
                \sigma_{\mathrm{point}}(\dot A_0) = \sigma_\infty \cup \sigma_0
            \end{equation}
            where
            \begin{equation}
                \sigma_\infty = \bigcup_{n=1}^{\infty}\bigcup_{m=-\infty}^{\infty}\{\lambda_{n,m}\},
            \end{equation}
            and $\{\lambda_{n,m}\}_{m\in\mathbb{Z}}$ for a fixed $n\in\mathbb{N}$ is a non-decreasing sequence of roots $\lambda\in\mathbb{R}$ of characteristic equation
            \begin{equation}
                \frac{\tan\left(a\sqrt{\frac{\lambda}{\epsilon_+}- (\frac{n\pi}{c})^2}\right)}{\epsilon_+\sqrt{\frac{\lambda}{\epsilon_+}- (\frac{n\pi}{c})^2}} = \frac{\tanh\left(b\sqrt{\frac{\lambda}{\epsilon_-} + (\frac{n\pi}{c})^2 }\right) }{\epsilon_- \sqrt{\frac{\lambda}{\epsilon_-} + (\frac{n\pi}{c})^2}}
                \label{eq:2d-nonhomog-eigenvalue}
            \end{equation}
            for $\lambda \ne \pm\epsilon_\pm(\frac{n\pi}{c})^2$ and we allow negative terms under square roots. If $\lambda=0$ is a solution, put $\lambda_{n,0}=0$, otherwise leave the index $\lambda_{n,0}$ undefined and put $\lambda_{n, \pm1}$ as the smallest positive, respectively largest negative, solution of the characteristic equation.

            There are eigenfunctions $f_{n,m}(x,y)$ corresponding to the eigenvalues as $f_{n,m}(x,y)=\mathcal{N}_{n,m}\cdot\phi_n(y)\psi_{n,m}(x)$ where $\phi_n(y)=\sqrt \frac{2}{c}\sin(\frac{n\pi}{c}y)$ and
            \begin{equation}
                \psi_{n,m}(x)=
                \begin{cases}
                    \sinh\left(\sqrt{\frac{\lambda_{n,m}}{\epsilon_-} + (\frac{n\pi}{c})^2}\ b\right)\sin\left(\sqrt{\frac{\lambda_{n,m}}{\epsilon_+} - (\frac{n\pi}{c})^2}\ (a-x)\right), & x\ge0 \\
                    \sin\left(\sqrt{\frac{\lambda_{n,m}}{\epsilon_+} - (\frac{n\pi}{c})^2}\ a\right)\sinh\left(\sqrt{\frac{\lambda_{n,m}}{\epsilon_-} + (\frac{n\pi}{c})^2}\ (b+x)\right), & x<0
                \end{cases}
                \label{eq:2d-nonhomog-eigenvector}
            \end{equation}
            and $\mathcal{N}_{n,m}$ is a norm constant.
            Further, define a set containing zero, one or two elements
            \begin{equation*}
                \sigma_0 \subseteq \{\nu_+, \nu_-\},
            \end{equation*}
            where $\nu_+=\epsilon_+(\frac{n\pi}{c})^2$ for at most one such $n\in \mathbb{N}$, which satisfies the equation
            \begin{equation}\label{eq:2d-nonhomog-singular-plus}
                \tanh\left(\frac{n\pi}{c}\sqrt{1+\frac{\epsilon_-}{\epsilon_+}}b\right)=\frac{\epsilon_-}{\epsilon_+}\frac{n\pi}{c}\sqrt{1+\frac{\epsilon_-}{\epsilon_+}}a,
            \end{equation}
            if it exists. Similarly, $\nu_-=-\epsilon_-(\frac{n\pi}{c})^2$ for at most one such $n\in \mathbb{N}$ satisfying
            \begin{equation}\label{eq:2d-nonhomog-singular-minus}
                \tanh\left(\frac{n\pi}{c}\sqrt{1+\frac{\epsilon_+}{\epsilon_-}}b\right)=\frac{\epsilon_+}{\epsilon_-}\frac{n\pi}{c}\sqrt{1+\frac{\epsilon_+}{\epsilon_-}}b,
            \end{equation}
            if it exists. It is possible that $\sigma_0$ is empty if neither of the solutions exist for a given parameter choice of $a,b,c$. To these potential eigenvalues $\nu_\pm$, the corresponding eigenfunction are given as $f_{\nu_\pm}(x,y)=\mathcal{N}_{n,m}^\pm\cdot\phi_n(y)\phi_{\nu_\pm}(x)$, where $n$ is the same as in definition of $\nu_\pm$ and $\phi_n$ is the normalised sine as specified above,
            \begin{equation}
                \phi_{\nu_+}(x) =
                \begin{cases}
                    (x-a)\sinh\left(b\frac{n\pi}{c}\sqrt{1+\frac{\epsilon_+}{\epsilon_-}}\right), & x\ge0, \\
                    -a\sinh\left(\frac{n\pi}{c}\sqrt{1+\frac{\epsilon_+}{\epsilon_-}}(x+b)\right), & x<0,
                \end{cases}
            \end{equation}
            \begin{equation}
                \phi_{\nu_-}(x) =
                \begin{cases}
                    b\sinh\left(\frac{n\pi}{c}\sqrt{1+\frac{\epsilon_-}{\epsilon_+}}(a-x)\right), & x\ge0, \\
                    (x+b)\sinh\left(a\frac{n\pi}{c}\sqrt{1+\frac{\epsilon_-}{\epsilon_+}}\right), & x<0.
                \end{cases}
            \end{equation}
            and $\mathcal{N}_{n,m}^\pm$ are again norm constants.
        \item Eigenvalues accumulate to $\pm\infty$.
        \item When $\kappa\not=1$, characteristic equation has no finite accumulation points.
        \item Zero is an infinitely-degenerate eigenvalue if and only if $\kappa=1$ and $a=b$.
    \end{enumerate}
    \label{prop:2d-spectrum}
\end{proposition}
\begin{proof}
    Begin by considering ansatz for the eigenfunctions in the form of 
    \begin{equation}\label{eq:ansatz-eigenvector}
        f_{n}(x,y)=\psi_n(x)\phi_n(y)
    \end{equation}
    for $\phi(y)=\frac{2}{c}\sin\left(\frac{n\pi}{c}\right)$, a Dirichlet basis in $L^2(J_2)$. In Section~\ref{section:essential-self-adjoint} we prove that all eigenvectors can be decomposed in this manner. There, we show that these eigenvectors of $\dot A_K$ form a complete orthonormal set in $L^2(\Omega_0,g)$.
    After substituting~\eqref{eq:ansatz-eigenvector} into eigenvalue equation for $\dot A_K$, we obtain two equations
\begin{align}
    \begin{split}
        \left(\psi_{n\pm}''(x)-\psi_{n\pm}(\frac{n\pi}{c})^2\right)\phi_n(y) = \mp\frac{\lambda}{\epsilon_\pm} \psi_{n\pm}(x)\phi_n(y)
        \label{eq:2d-eigensum}
    \end{split}
\end{align}
where $\psi_{n\pm}=\restr{\psi_n}{\Omega_\pm}$. The domain of the operator further restrain possible choice of $\psi_n$ as
\begin{equation}\label{eq:separated-domA-conditions}
    \psi_{n+}(a) = 0 = \psi_{n-}(-b), \quad \psi_{n+}(0)=\psi_{n-}(0), \quad \epsilon_+\psi_{n+}'(0) = -\epsilon_-\psi_{n-}'(0).
\end{equation}
By further solving the differential equations in terms of exponentials and subjecting the solutions to boundary and interface conditions~\eqref{eq:separated-domA-conditions}, we obtain for $\lambda\not=\pm\epsilon_\pm \left(\frac{n\pi}{c}\right)^2$ characteristic equation~\eqref{eq:2d-nonhomog-eigenvalue}. For the possible singular case, solution in $\Omega_\pm$ is a linear function and an exponential and there is~\eqref{eq:2d-nonhomog-singular-plus}--\eqref{eq:2d-nonhomog-singular-minus}.
Note that the last equations can be satisfied for certain parameters. For example, choice $a=1, b=2, c\approx4.64, \epsilon_+=\epsilon_-=1, n=1$ satisfies equation~\eqref{eq:2d-nonhomog-singular-minus}.

\begin{figure}
    \centering
    \includegraphics[width=.7\linewidth]{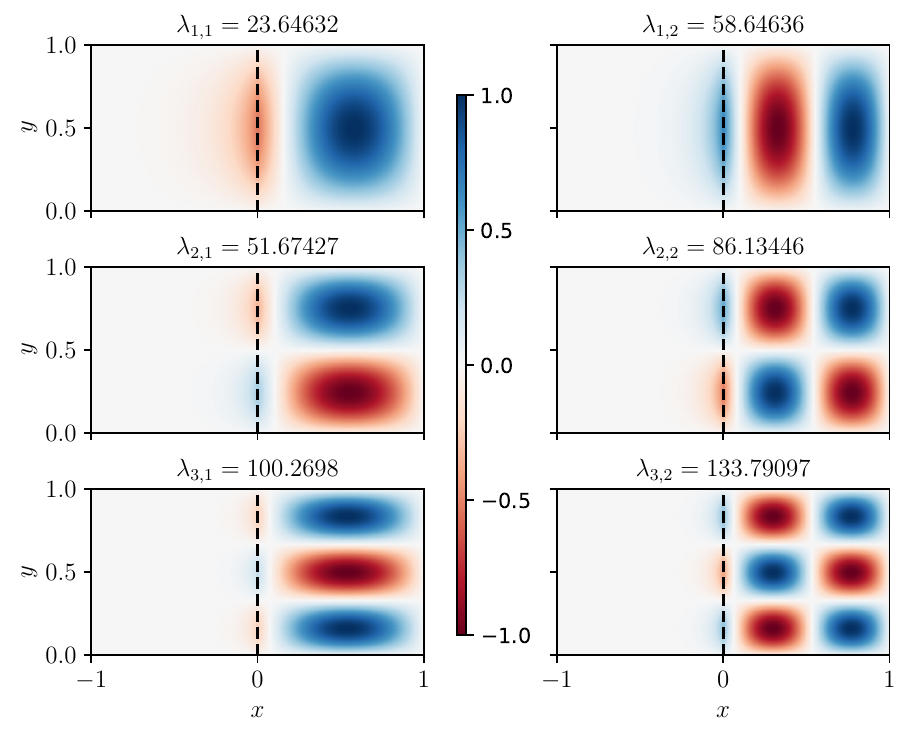}
    \caption{Eigenvectors $f_{n,m}$ in \eqref{eq:2d-nonhomog-eigenvector} plotted such that $\|f_{n,m}\|_\infty=1$, $\epsilon_+=\epsilon_-=a=b=c=1$}
    \label{fig:plt2df_a1b1p1m1}
\end{figure}
\begin{figure}
    \centering
    \includegraphics[width=.7\linewidth]{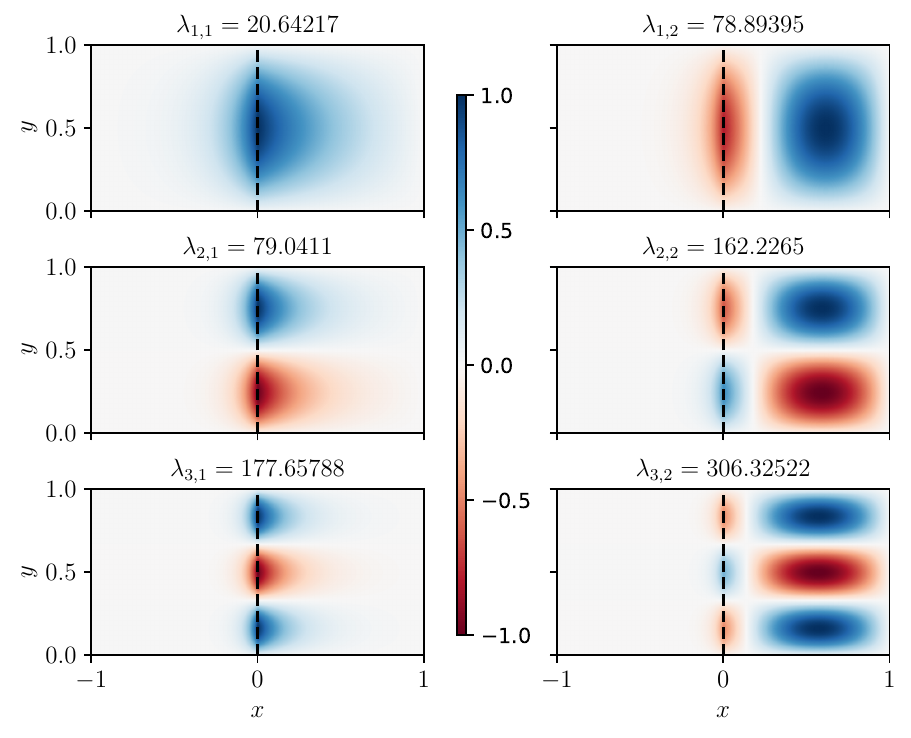}
    \caption{Eigenvectors $f_{n,m}$ in \eqref{eq:2d-nonhomog-eigenvector} plotted such that $\|f_{n,m}\|_\infty=1$, $\epsilon_+=3$, $\epsilon_-=a=b=c=1$}
    \label{fig:plt2df_a1b1p3m1}
\end{figure}

The second point of the statement is implied by the following. The function $x\mapsto\frac{\tanh x}{x}$ is bounded for real arguments. Function $g(x):= \frac{\tan x}{x}$ is oscillating and during each period of $\tan x$, the function $g$ has range $\mathbb{R}$. Thus, there are infinitely many eigenvalues which accumulate to $\pm\infty$. 

    Let $\Lambda\in\mathbb{R}$ be a finite accumulation point of roots of the characteristic equation. Then there exists some sequence of eigenvalues $\{\lambda_n\}_{n=1}^\infty\subset\sigma(\dot A_0)$ which converges to $\Lambda$, i.e. $\lim_{n\to+\infty}\lambda_n=\Lambda$. Without loss of generality, assume $b>a$. As the limit value is finite, there exists $n_0\in\mathbb{N}$ such that for all $n>n_0$, the value $\lambda_n$ lies in an interval $(-\epsilon_-(\frac{n\pi}{c})^2,\epsilon_+(\frac{n\pi}{c})^2)$. We will rearrange characteristic equation~\eqref{eq:K0-eigenequation} to form
    \begin{equation}
        \frac{\tanh(a\sqrt{(\frac{n\pi}{c})^2-\frac{\lambda_n}{\epsilon_+})}}{\tanh(b\sqrt{(\frac{n\pi}{c})^2 + \frac{\lambda_n}{\epsilon_-}})} = \frac{\epsilon_+}{\epsilon_-}\sqrt{\frac{(\frac{n\pi}{c})^2 - \frac{\lambda_n}{\epsilon_+}}{(\frac{n\pi}{c})^2 + \frac{\lambda_n}{\epsilon_-}}}
    \end{equation}
    and take limit $n \to +\infty$ ($\lambda_n\to \Lambda < +\infty$) on both sides of the equation. This reduces to the necessary condition 
    \begin{equation}
        1=\frac{\epsilon_+}{\epsilon_-}.
    \end{equation}

As can be seen from~\eqref{eq:2d-nonhomog-eigenvalue} by substituting $\lambda=0$, it satisfies the equation iff $\kappa=1$ and $a=b$ --- as is obvious from the fact that  $\frac{\tanh\left(a x\right)}{\tanh \left(b x\right)}$ is not a constant in $x$ unless $a=b$.
\end{proof}
\begin{proposition}
    Fix curvature $K=\pm1$. The eigenvalues of operator $\dot A_K$ are the solutions $\lambda$ to the characteristic equation
\begin{equation}
    \begin{vmatrix}
        \psi_1^+(a) & \psi_2^+(a) & 0 & 0 \\
        0 & 0 & \psi_1^-(-b) & \psi_2^-(-b) \\
        \psi_1^+(0) & \psi_2^+(0) & -\psi_1^-(0) & - \psi_2^-(0) \\
        \epsilon_+\psi_1^{+\prime}(0) & \epsilon_+\psi_2^{+\prime}(0) & \epsilon_-\psi_1^{-\prime}(0) & \epsilon_-\psi_2^{-\prime}(0)
    \end{vmatrix}
    =0,
    \label{eq:manifold-det}
\end{equation}
for each $n\in\mathbb{N}$ where the $\psi_\iota^+: [0,a]\to\mathbb{R}$ and $\psi_\iota^-: [-b,0]\to\mathbb{R}$ are defined as in Table~\ref{table:solution-pq} and are implicitly dependent on $n$ and $\lambda$ via $\mu$ and $\nu$. In case there are multiple solutions of~\eqref{eq:manifold-det} for a fixed $n\in\mathbb{N}$, we denote the solutions as $\lambda_{n,k}$, $k\in\mathbb{Z}$. This is well-defined notation as self-adjoint operators in separable Hilbert spaces have at most countable point spectrum.

 Finally, we have
  \begin{equation}
      \overline{\sigma_{\mathrm{point}}(\dot A_K)} = \bigcup_{n=1}^\infty \bigcup_{k\in\mathbb{Z}} \ \{\lambda_{n,k}\}
 \end{equation}
 with $\lambda_{n,k}$ being solutions of~\eqref{eq:manifold-det}, for fixed $n\in\mathbb{N}$ sorted in an increasing manner as $\lambda_{n,k}<\lambda_{n,k+1}$ for all $k \in \mathbb{Z}$. The resulting eigenfunctions are of form $\psi_{n,k}(x,y)=\phi_{n,k}(x)\sin\left(\frac{n\pi}{c}\right)$ for
\begin{equation}
    \phi_{n,k}(x):=
    \begin{cases}
        C_1^+\psi_1^+(x) + C_2^+\psi_2^+(x),&\quad x\geq 0,\\
        C_1^-\psi_1^-(x) + C_2^-\psi_2^-(x),&\quad x\leq 0,
    \end{cases}
\end{equation}
for constants $C_1^\pm$, $C_2^\pm$ determined up to the same multiplicative factor using procedure in the reference.
\end{proposition}

\begin{proof}
    Straightforward by applying interface and boundary conditions on separated ansatz as in previous proof, leading to Legendre differential equation
\begin{equation}\label{eq:legendreeq}
    \left(1-x^{2}\right)y''-2xy'+\left(\nu (\nu +1)-{\frac {\mu^{2}}{1-x^{2}}}\right)y=0.
\end{equation}
\end{proof}
\begin{table}[H]
    \centering
    \begin{tabular}{r | l l | l l}
        \toprule
        & $\psi_1^\pm(x)$ 
        & $\psi_2^\pm(x)$ 
        & $\mu$ 
        & $\nu_\pm$
        \\
        \midrule
        $K=+1$ 
        & $P_{\nu\pm}^{\mu}(\sin x)$ 
        & $Q_{\nu\pm}^{\mu}(\sin x)$ 
        & $\frac{n\pi}{c}$ 
        & $\frac{1}{2}(\sqrt{1\pm\frac{4\lambda_{n,k}}{\epsilon_\pm}}-1)$
        \\
        $K=-1$
        & $\frac{P_{\nu\pm}^{\mu}(\tanh x)}{\sqrt{\cosh x}}$ 
        & $\frac{Q_{\nu\pm}^{\mu}(\tanh x)}{\sqrt{\cosh x}}$ 
        & $-\frac{1}{2} + \iu \frac{n\pi}{c}$ 
        & $\frac{1}{2}\sqrt{1\mp\frac{4\lambda_{n,k}}{\epsilon_\pm}}$
        \\
        \bottomrule
    \end{tabular}
    \caption{Choice of eigenfunctions on $\Omega_+$ and $\Omega_-$ for fixed $n\in\mathbb{N}$. Eigenvalue $\lambda_{n,k}$ is a solution to~\eqref{eq:manifold-det}. $P^{\mu}_\nu$ and $Q^{\mu}_\nu$ are associated linearly independent Legendre function of first and second kind, respectively.}
    \label{table:solution-pq}
\end{table}

\section{Essential self-adjointness via eigenfunctions}\label{section:essential-self-adjoint}
In this section, we show that $\dot A_K$ is an essentially self-adjoint operator regardless of value of contrast $\kappa$ and constant Gaussian curvature $K\in\mathbb{R}$ of the ambient Riemannian 2-manifold.
\subsection{Separation of variables}
As all coefficients in $\dot A_K$ and metric $g$ are independent of second variable $x_2\equiv y$, it is beneficial to decompose the operator into a direct sum of one-dimensional operators via separation of variables by utilising an orthonormal basis of $L^2(J_2)$ such that it satisfies Dirichlet condition on $J_2$.
\begin{lemma}[\cite{krejcirik2010ptsymmetric}]
    \label{lemma:separation}
    \begin{equation}
        \forall \Psi \in L^2(\Omega_0,g), \quad \Psi(x_1,x_2)=\sum_{n=0}^{+\infty}\psi_n(x_1)\phi_n(x_2)\quad \mathrm{in}\ L^2(\Omega_0,g),
    \end{equation}
    where
    \begin{equation}
        \phi_0(x_2)=\frac{1}{\sqrt{c}},\quad\phi_n(x_2) = \sqrt\frac{2}{c}\sin\left(\frac{n\pi}{c}x_2\right),\quad \psi_n(x_1)=\left(\phi_n,\Psi(x_1,\cdot)\right)_{L^2(J_2)}.
    \end{equation}
\end{lemma}

\begin{lemma}\label{lemma:AK-separation}
    The operator $\dot A_K$ can be decomposed for $\Psi\in\dom \dot A_K$, $\Psi(x,y)=\psi_n(x)\phi_n(y)$ as 
    \begin{equation}
        \dot A_K\Psi = (A_K^n\psi_n)\otimes\phi_n
    \end{equation}
    where $\phi_n$, $\psi_n$ are given in Lemma~\ref{lemma:separation}, $A_K^n\in \mathcal{L}\left(L^2(J_1, \dd \mu_K)\right)$, $n\in\mathbb{N}_0$ denotes
\begin{equation}
    A_K^n:=\begin{pmatrix}\epsilon_+\\-\epsilon_-\end{pmatrix}\cdot
    \begin{cases}
        -\pp_1^2 + \tan(x_1)\pp_1 + \frac{(n\pi)^2}{c^2\cos^2(x_1)}, & \mathrm{if}\  K=1, \\
        -\pp_1^2+(\frac{n\pi}{c})^2, & \mathrm{if}\ K=0, \\
        -\pp_1^2 - \tanh(x_1)\pp_1 + \frac{(n\pi)^2}{c^2\cosh^2(x_1)}, & \mathrm{if}\ K=-1
    \end{cases}
\end{equation}
and $A_K^n$ acts on $L^2(J_1,\dd\mu_K)$ with measure
\begin{equation}
    \dd\mu_K:=
    \begin{cases}
        \cos(x_1)\dd x_1, & \mathrm{if} \ K=+1, \\
        \dd x_1, &\mathrm{if} \ K=0, \\
        \cosh(x_1) \dd x_1, & \mathrm{if} \ K=-1
    \end{cases}
\end{equation}
and domain of $A_K^n$ is given by
\begin{align}
    \begin{split}
        \dom&\ A_K^n:=\\
    &\left\{
        \begin{array}{l|l}
         &\psi_+(a)=\psi_-(-b)=0, \\
        \psi=\begin{pmatrix}\psi_+\\\psi_-\end{pmatrix} \in H^2 \left( \left(-b,0\right),\dd\mu_K\right)\oplus H^2\left( \left(0,a\right),\dd\mu_K\right) &\psi_+(0) = \psi_-(0)\\
        &\epsilon_+ \psi_+'(0_+) = -\epsilon_- \psi_-'(0_-)
        \end{array}
    \right\}.
    \end{split}
\end{align}

\end{lemma}
\begin{proof}
    The action of $A_K^n$,  $n\in\mathbb{N}_0$ is straightforward. Its domain is projection to one-dimensional $L^2(J_1,\dd\mu_K)$ space.
\end{proof}

\subsection{Unitary transform}
From now onwards, we will use notation $x_1\equiv x$,  $x_2\equiv y$. We will use the separation of variables to quantify the effect of curvature as an effective potential.
\begin{lemma}\label{lemma:unitary-K1}
    Unitary transform $U_{+1}: L^2(J_1,\dd x) \to L^2(J_1, \dd \mu_{+1})$
    \begin{equation}
        (U_{+1}\psi)(x):=\cos(x)^{-\frac{1}{2}}\psi(x)
    \end{equation}
    transforms operator $A_{+1}^n$, $n\in\mathbb{N}$, to a ``flat'' operator plus a discontinuous potential given by
    \begin{align}
        \begin{split}
            U_{+1}^{-1}A_{+1}^n U_{+1} &=A_0^0 + V_{+1}^n,\\
            (V_{+1}^n\psi)(x) &= \frac{8\left(\frac{n\pi}{c}\right)^2-3-\cos{2x}}{8\cos^2{x}}\begin{pmatrix}\epsilon_+\psi_+(x)\\-\epsilon_-\psi_-(x)\end{pmatrix},
        \end{split}
    \end{align}
    where $A_0^0$ (defined in Lemma~\ref{lemma:AK-separation}) is operator corresponding to composite (metamaterial \& material) one-dimensional string analogue in Euclidean space.
\end{lemma}
\begin{proof}
    The proof is adapted to our case based on~\cite[Lemma 4.3]{krejcirik2010ptsymmetric}. $U_{+1}$ is regular and unitary as $x\in(-b,a)\subset(-\frac{\pi}{2}, \frac{\pi}{2})$ and $(U_{+1}\psi, U_{+1}\phi)_{(J_1, \dd \mu_{+1})} = \int_{J_1}(\cos^{-\frac{1}{2}})^2 \overline{\psi(x)}\phi(x) \dd \mu_{+1} = \int_{J_1}\overline{\psi}\phi \dd x=(\psi, \phi)_{(J_1, \dd x)}$.
    It is straightforward to verify that the boundary condition are satisfied and ultimately, $\dom A_{0}^0 =\dom U_{+1}^{-1}A_{+1}^mU_{+1}$.
\end{proof}

\begin{lemma}\label{lemma:unitary-K-1}
    Unitary transform $U_{-1}: L^2(J_1,\dd x) \to L^2(J_1, \dd \mu_{-1})$
    \begin{equation}
        (U_{-1}\psi)(x):=\cosh(x)^{-\frac{1}{2}}\psi(x)
    \end{equation}
    transforms $A_{-1}^n$ to
    \begin{align}
        \begin{split}
            U_{-1}^{-1}A_{-1}^n U_{-1} &=A_{0}^0 + V_{-1}^n,\\
            (V_{-1}^n\psi)(x) &= \frac{8\left(\frac{n\pi}{c}\right)^2+3+\cosh{2x}}{8\cosh^2{x}}\begin{pmatrix}\epsilon_+\psi_+(x)\\-\epsilon_-\psi_-(x)\end{pmatrix}.
        \end{split}
    \end{align}
\end{lemma}
\begin{proof}
    Analogous to Lemma~\ref{lemma:unitary-K1}.
\end{proof}

\begin{remark}\label{remark:general-unitary}
    Note that the transformation from curved to \emph{Euclidean} case plus potential can be made in general. For linear operator $H$ in $L^2(\Omega_0, f \dd x \dd y)$ given by $H = -g^{-\frac{1}{2}}\pp_i g^{\frac{1}{2}} g^{ij} \pp_j$ and metric tensor $(g_{ij}):=\mathrm{diag}(1, f^2)$, the sought transformation is
\begin{equation}
    (\hat H\psi)(x):= (g^{\frac{1}{4}}H g^{-\frac{1}{4}}\psi)(x) = -(\pp_i g^{ij} \pp_j \psi)(x) + V(x)\psi(x)
\end{equation}
in $L^2(\Omega)$, where $V$ is fully determined by $g$. The construction in applied to those Riemannian 2-manifolds such that, in normal coordinates, $g(x,y)=\mathrm{diag}(1,f^2(x))$, the unitary transformation reads
  \begin{align}
    \begin{split}
      U&: L^2(J_1,\dd x)\to L^2(J_1,\dd \mu_K)\\
      (U\psi)(x)&:=f^{-1/2}(x)\psi(x)
    \end{split}
  \end{align}
  and then the potential is given as
  \begin{equation}
    V_K^n(x)=\frac{f''(x)}{2 f(x)}-\frac{f'(x)^2}{4 f(x)^{2}}+\frac{\left(\frac{n\pi}{c}\right)^2}{f(x)^{2}}.
  \end{equation}
  It should be noted that the Jacobi equation connecting the metric and curvature is now more complicated.
\end{remark}

\subsection{Essential self-adjointness}
\begin{lemma}\label{lemma:A00-self-adjoint}
    Fix $n\in\mathbb{N}_0$. The operator $A_0^n$ in zero curvature is self-adjoint. Let $S_n:=\{f_{n,m}\}_{m\in\mathbb{Z}}$ be eigenvectors corresponding to $A_0^n$. Then $S_n$ forms an orthonormal basis of $L^2(J_1)$.
\end{lemma}
\begin{proof}
    In~\cite{hussein2014sign}, it is shown that $A_0^0$ is self-adjoint and $\sess(A_0^0)=\emptyset$. As the difference between $A_0^n$ and $A_0^0$ is a constant real potential, hence symmetric and $A_0^0$-bounded, $A_0^n=A_0^0 + (n\pi/c)^2$, the operator $A_0^n$ is self-adjoint and its essential spectrum is also empty~\cite{kato2013perturbation}. Hence, spectral theorem for self-adjoint operators together with empty essential spectrum implies~\cite{davies1996spectral} the second statement.
\end{proof}
\begin{lemma}\label{lemma:AKm-self-adjoint}
    Fix $K\in\mathbb{R}$. Operator $A^n_K$ is self-adjoint for all $n\in \mathbb{N}_0$. Eigenvectors of $A_K^n$ form a complete orthonormal set of $L^2(J_1,\dd\mu_K)$.
\end{lemma}
\begin{proof}
    Case of $K=0$ is the previous lemma, the other case can be reduced to $K=\pm1$. Due to unitary transformations in Lemmata~\ref{lemma:unitary-K1} and \ref{lemma:unitary-K-1} and self-adjointness of $A_0^0$ in Lemma~\ref{lemma:A00-self-adjoint}, it is straightforward to show the statement. Unitary transformations preserve self-adjointness and the potentials $V_{\pm 1}^n$ are real and bounded.
\end{proof}
\begin{remark}
    It can easily be seen that the same argument applies for setup with more general metric in Fermi coordinates given by $g(x,y)=\mathrm{diag}(1, f^2(x))$ with $f$,  $f^{-1}$,  $f'$,  $f'' \in L^\infty(J_1)$, by Remark~\ref{remark:general-unitary}.
\end{remark}

\begin{theorem}[\cite{davies1996spectral}]\label{thm:davies-ess-self-adjoint}
Let $H$ be a symmetric operator on a Hilbert space $\mathcal{H}$ with domain $L$, and let $\{f_n\}_{n=0}^\infty$ be a complete orthonormal set in $\mathcal{H}$. If each $f_n$ lies in $L$ and there exist $\lambda_n\in\mathbb{R}$ such that $H f_n = \lambda_n f_n$ for every $n$, then $H$ is essentially self-adjoint. Moreover, the spectrum of H is the closure in $\mathbb{R}$ of the set of all $\lambda_n$.
\end{theorem}

\begin{theorem}\label{theorem:dAK-essential}
    Operator $\dot A_{K}$ given in~\eqref{eq:def-dotA_K} satisfies the following:
    \begin{enumerate}
        \item its eigenfunctions form a complete orthonormal set in $L^2(\Omega_0, g)$,
        \item $\overline{\sigma_\mathrm{point}(\dot A_K)}=\cup_{n\in\mathbb{N}}\ \sigma_{\mathrm{point}}\left(A_K^n\right)$,
        \item it is essentially self-adjoint, ie. there exists unique self-adjoint $A_K:=\overline{\dot A_K}$ for arbitrary constant Gaussian curvature $K\in\mathbb{R}$ and contrast $\kappa>0$,
        \item $\sigma(A_K)=\overline{\sigma_\mathrm{point}(\dot A_K)}$.
        \item operator $A_K$ can be expressed as a direct sum of transversal operators
            \begin{equation}
                A_K = \bigoplus_{n=1}^\infty \left(\mathbb{I}\otimes A_K^n\right) P_n,
            \end{equation}
            where $P_n: L^2(\Omega_0,g)\to L^2(\Omega_0,g)$ are rank-one operators; for $\Psi\in L^2(\Omega_0,g)$
            \begin{equation}
                (P_n\Psi)(x,y):=\left(\phi_m, \Psi(x,\cdot)\right)_{L^2(J_2)}\phi_n(y),\quad \phi_n(y):=\sqrt{\frac{2}{c}}\sin\left(\frac{n\pi}{c}y\right)
            \end{equation}
            and transversal (self-adjoint) operator $A_K^n$ was already given in Lemma~\ref{lemma:AK-separation}.
    \end{enumerate}
\end{theorem}
\begin{proof}
    Denote $\{\psi_{n,m}\}_{m\in\mathbb{Z}}$ the eigenfunctions of $A_K^n$ for some $n\in\mathbb{N}$, forming a complete orthonormal set of $L^2(J_1,\dd\mu_K)$ due to Lemma~\ref{lemma:AKm-self-adjoint}. Set $\{\psi_{n,m}\otimes\phi_n\}_{n,m\in\mathbb{Z}}$ is then a complete orthonormal set in $L^2(\Omega_0,g)$ (\cite[V, Example 1.10]{kato2013perturbation}) as $\{\phi_n\}_{n\in\mathbb{N}}$ is a basis in $L^2(J_2)$ (in notation of Lemma~\ref{lemma:separation}). As 
    \begin{equation}
        \dot A_K \left(\psi_{n,m}\otimes\phi_n\right) = \left(A_K^n\psi_{n,m}\right)\otimes\phi_n,
    \end{equation}
    function $\psi_{n,m}\otimes\phi_n$ is an eigenfunction of $\dot A_K$ for each $n\in\mathbb{N}$, $m\in\mathbb{Z}$. Hence, Theorem~\ref{thm:davies-ess-self-adjoint} applies and $\dot A_K$ is essentially self-adjoint. Also, as $\{\psi_{n,m}\otimes\phi_n\}_{n\in\mathbb{N},m\in\mathbb{Z}}$ is complete, these are precisely all eigenfunctions of $\dot A_K$.
    
    The last point follows from proof of Theorem~\ref{thm:davies-ess-self-adjoint} and the following. Let $\Psi\in\dom A_K\subset L^2(\Omega_0,g)$. By Lemma~\ref{lemma:separation}, it is the limit of $\{\Psi_N\}_{N\in\mathbb{N}}$, $\Psi_N:=\sum_{n=1}^N \psi_n\otimes\phi_n$ with notation from the Lemma. As eigenvectors of $A_K$ form an ortonormal basis of  $L^2(\Omega,g)$, proof of Theorem~\ref{thm:davies-ess-self-adjoint} implies that $\lim_{N\to\infty}A_K \Psi_N=A \Psi$. Also, for $k\in\mathbb{N}$,
    \begin{align}
        \begin{split}
            \left(\phi_k,A_K\Psi_N\right)_{L^2(J_2)} = \sum_{n=1}^N\left(\phi_k,A_K^n\psi_n\otimes\phi_n\right)_{L^2(J_2)} = A_K^k\psi_k
        \end{split}
    \end{align}
    and hence, $A_K\Psi=\sum_{n=1}^\infty \left(A_K^n\psi_n\right)\otimes\psi_n$. The domain of $A_K^n$ is obtained by projection and characterization of Sobolev spaces in one dimension.
\end{proof}

\section{Form approach for non-critical contrast}\label{section:form-approach}
In Section~\ref{section:essential-self-adjoint}, we introduced a selfajoint operator $A_K$. In the following section, we will show that an alternative realisation $\mathcal{A}_K$ of the self-adjoint operator corresponding to $\dot A_K$ is possible for contrast $\kappa\not=1$ --- although we prove that the different techniques used to introduce the self-adjoint realisations result in the same operator $A_K=\mathcal{A}_K$ for $\kappa\not=1$. Furthermore, we show that the operator has compact resolvent for $\kappa\not=1$ and hence empty essential spectrum.

Let us recall the representation theorems on which this section is based on. Recall that continuity of a sequilinear form $a: \mathcal{V}\times\mathcal{V}\to\mathbb{C}$ in a Hilbert space $\mathcal{V}$ means that for some $C>0$:
\begin{equation}
    |a(u,v)|\leq C||u||_\mathcal{V}||v||_\mathcal{V},\quad \forall u,v\in\mathcal{V}.
\end{equation}

\begin{theorem}[{\cite[Thm. 2.1]{almog2015spectrum}}]\label{thm:almog1}
    Let $\mathcal{V}$ denote a Hilbert space. Let $a$ be a continuous sesquilinear form on $\mathcal{V}\times\mathcal{V}$. If  $a$ satisfies, for some $\Phi_1,\Phi_2\in\mathcal{L}(\mathcal{V})$,
    \begin{align}\label{eq:almog-twocond}
        \begin{split}
            |a(u,u)|+|a(u,\Phi_1u)|\geq\alpha||u||_{\mathcal{V}}^2,&\quad \forall u\in\mathcal{V},\\
            |a(u,u)|+|a(\Phi_2u,u)|\geq\alpha||u||_\mathcal{V}^2,&\quad\forall u\in\mathcal{V},
        \end{split}
    \end{align}
    then $\mathcal{A}\in\mathcal{L}(\mathcal{V})$ defined via
\begin{equation}
    a(u,v)=\left(u,\mathcal{A}v\right)_\mathcal{V}
\end{equation}
is a continuous isomorphism from $\mathcal{V}$ onto  $\mathcal{V}$. Moreover,  $\mathcal{A}^{-1}$ is continuous.
\end{theorem}
Now, consider two Hilbert spaces $\mathcal{V}$ and $\mathcal{H}$ such that
\begin{align}\label{eq:almog-HV}
    \begin{split}
        \mathcal{V}\subset \mathcal{H}&,\quad \mathcal{V}\ \mathrm{is\ dense\ in}\ \mathcal{H},\\
        \forall u\in\mathcal{V}&:\quad||u||_\mathcal{H}\leq C||u||_\mathcal{V}
    \end{split}
\end{align}
for some $C>0$.
\begin{theorem}[{\cite[Thm. 2.2]{almog2015spectrum}}]\label{thm:almog2}
    Let $a$ be a continuous sesquilinear form satisfying~\eqref{eq:almog-twocond}. Let $\mathcal{H}\supset\mathcal{V}$ be a Hilbert space and suppose that~\eqref{eq:almog-HV} holds for the Hilbert spaces $\mathcal{V}$, $\mathcal{H}$. Further assume that $\Phi_1$,  $\Phi_2$ extend to a continuous linear map also in $\mathcal{H}$. Define operator $S:\dom S\to\mathcal{H}$ using
    \begin{align}
        \begin{split}
            \dom S&:=\left\{v\in\mathcal{V}: u\mapsto a(u,v) \ \mathrm{is\ continuous\ on}\ \mathcal{V}\ \mathrm{in\ the\ norm\ of\ } \mathcal{H}\right\},\\
            a(u,v)&=:\left(u,Sv\right)_\mathcal{H},\quad \forall v\in\dom S,\ \forall u\in\mathcal{V}.
        \end{split}
    \end{align}
    Then,
    \begin{enumerate}
        \item $\dom S$ is dense in both $\mathcal{V}$ and $\mathcal{H}$,
        \item $S$ is closed,
        \item $S$ is bijective from $\dom S$ onto $\mathcal{H}$ and $S^{-1}\in\mathcal{L}(\mathcal{H})$.
        \item Let $b$ denote the conjugate sesquilinear form of $a$ given by
             \begin{equation*}
                 (u,v)\mapsto b(u,v):=\overline{a(v,u)}.
            \end{equation*}
          and denote $\tilde S$ the operator associated to $b$ by the same construction --- then
          \begin{equation*}
              S^*=(\tilde S)^* \ \mathrm{and}\ (\tilde S)^*=S.
          \end{equation*}
    \end{enumerate}
\end{theorem}
Application of this representation theorem have been used in context of Schr\"odinger operators with complex potentials, for example in~\cite{krejcirik2017non},~\cite{arnal2023resolvent} and~\cite{semoradova2022diverging}.

\subsection{Space and operator definition}
The notation of Hilbert spaces on manifold is the same as in Section~\ref{section:setting}. We will use the following notation for function restrictions
\begin{align}
    \begin{split}
        u_\pm &:= \restr{u}{\Omega\pm}\quad \mathrm{for}\ u\in L^2(\Omega_0,g),\\
        ||u_\pm|| &:= ||u_\pm||_{L^2(\Omega_{\pm}, \restr{g}{\Omega_{\pm}})}.
    \end{split}
\end{align}
Scalar product $(\cdot, \cdot)$ will be linear in second argument and that of $L^2(\Omega_0,g)$, unless specified otherwise.
Start by defining Sobolev spaces of functions vanishing on the boundary of $\Omega$ restricted to $\Omega_\pm$ as
\begin{equation}
    H^1_{0,\mathcal{C}}(\Omega_\pm,g):=\left\{ \restr{f}{\Omega_\pm}: f\in H^1_0(\Omega_0,g) \right\}.
\end{equation}

Further, define an even cut-off function $\xi: \Omega_0\to\mathbb{R}$, $\xi\in C_0^\infty(\Omega_0)$ for $x_0, x_1\in\mathbb{R}$, $x_0<x_1<a$, $x_0<x_1<b$ such that $\xi(-x,y)=\xi(x,y)$ for all $(x,y)\in\Omega_0$ and, in addition, satisfying
\begin{equation}
    \xi(x,y)=
    \begin{cases}
        1, & x\in(-x_0,x_0),\\
        0, & x\in(-b,a)\setminus(-x_1,x_1).
    \end{cases}
\end{equation}
In interval $(-x_1,-x_0)\cup(x_0,x_1)$ it is defined such that $\xi\in C_0^\infty(\Omega_0)$ and $||\xi||_{L^2(\Omega_\pm,g)}=1$. Define mirroring operator 
\begin{align}
    \begin{split}
        P_\pm: H^1_{0,\mathcal{C}}(\Omega_\pm)&\to H^1_{0,\mathcal{C}}(\Omega_\mp),\\
        (P_\pm u)(x,y)=u(-x,y)\ &\mathrm{for}\ x\in (-x_1,x_1),\ y\in(0,c).
    \end{split}
\end{align}

In further analysis, we use it with a multiplication factor in form of a cut-off function $\xi$ identically zero outside of $(-x_1,x_1)$. So formally, we rather use the operator $\xi P_\pm:H^1_0(\Omega_\pm)\to H^1_0(\Omega_\mp)$ which is properly defined on functions over the whole of $\Omega_0$. 

Define transforms $T_\iota: H^1_0(\Omega_0,g)\to H^1_0(\Omega_0,g)$ for $u\in H^1_0(\Omega_0,g)$ using
\begin{align}
    \label{eq:T-transform}
    \begin{split}
        T_1 u:=
        \begin{cases}
            u_+ &\mathrm{in}\ \Omega_+\\
            -u_- + 2 R_+u_+ &\mathrm{in}\ \Omega_-
        \end{cases}
    \end{split},
    \begin{split}
        T_2 u:=
        \begin{cases}
            -u_+ + 2 R_-u_- &\mathrm{in}\ \Omega_- \\
            u_- &\mathrm{in}\ \Omega_-
        \end{cases}
    \end{split}
\end{align}
with $R_\pm: H^1_{0,\mathcal{C}}(\Omega_\pm)\to H^1_{0,\mathcal{C}}(\Omega_\mp)$, $R_\pm=\xi P_\pm$. Operators $T_\iota$ are bounded as operators acting in $H^1_0(\Omega_0,g)$ as
\begin{align}
    \begin{split}
         (T_\iota u)_+|_\mathcal{C}=(T_\iota u)_-|_\mathcal{C}, \quad \mathrm{for}\ u=(u_+,u_-)\in H^1_{0,\mathcal{C}}(\Omega_+)\oplus H^1_{0,\mathcal{C}}(\Omega_-)
    \end{split}
\end{align}
(in sense of traces) due to $(R_\pm u_\pm)|_\mathcal{C}=u_\pm|_\mathcal{C}$ for the same $u$.

For application of generalized Lax-Milgram lemma~\cite[Theorem 2.2]{almog2015spectrum}, let us introduce the sesquilinear form associated to operator $\dot A_K$ defined in~\eqref{eq:def-dotA_K} given as 
\begin{align}
    \begin{split}
        \dot a&: L^2(\Omega_0,g) \times \dom \dot A_K \to L^2(\Omega_0,g),\\
        \dot a(u,v)&:=(u,\dot A_K v).
    \end{split}
\end{align}
By invoking standard density arguments of test functions $C_0^\infty(\Omega_0)$ in  $L^2(\Omega_0,g)$ and $\dom\dot A_K$ and definition of weak derivatives, its domain can be augmented and it is given, for $u$, $v\in H_0^1(\Omega_0,g)$, by
\begin{align}\label{eq:form-a}
    \begin{split}
        a&: H_0^1(\Omega_0,g) \times H_0^1(\Omega_0,g) \to L^2(\Omega_0,g),\\
        a(u,v)&:=(\nabla u,\epsilon \nabla v)=\epsilon_+ \int_{\Omega_+} \overline{\nabla u}\nabla v - \epsilon_- \int_{\Omega_-}\overline{\nabla u}\nabla v
    \end{split}
\end{align}
and the derivatives are understood in the weak sense.

\subsection{Results}
\begin{theorem}\label{theorem:form-A-self-adjoint}
    Let $(\mathcal{M},g)$ be a Riemannian manifold with constant Gaussian curvature $K\in\mathbb{R}$ and $\Omega_0$, $\mathcal{C}$ as introduced in Section~\ref{section:setting}. Then, for contrast $\kappa=\frac{\epsilon_+}{\epsilon_-}\not = 1$, there is a unique self-adjoint operator $\mathcal{A}_K: \dom  \mathcal{A}_K \subset L^2(\Omega_0,g)\to L^2(\Omega_0,g)$ associated to the sesquilinear form $a$ given in~\eqref{eq:form-a} by
    \begin{equation}
        a(u,v)=:(u,\mathcal{A}_K v)_g, \quad u\in H_0^1(\Omega_0,g),\ v\in\dom \mathcal{A}_K\subset H_0^1(\Omega_0,g)
    \end{equation}
    with domain 
    \begin{equation}
        \dom  \mathcal{A}_K:= \left\{v\in H_0^1(\Omega_0,g): \Delta v_\pm \in L^2(\Omega_\pm,g),\ \restr{\left(\epsilon_+\partial_x v_+ +\epsilon_-\partial_x v_-\right)}{\mathcal{C}} =0\right\}
    \end{equation}
    where the interface condition is understood in the weak sense of 
    \begin{align}
        \begin{split}
            \Delta v_\pm\in &L^2(\Omega_\pm,g),\ \restr{\left(\epsilon_+\partial_x v_+ +\epsilon_-\partial_x v_-\right)}{\mathcal{C}} =0 \\
                &\vcentcolon\iff \forall u\in H_0^1(\Omega_0,g):\ \int_{\Omega_0} \overline{\nabla u}\epsilon\nabla v = - \int_{\Omega_0} \overline{u} \nabla\cdot\left(\epsilon\nabla v\right)
        \end{split}
    \end{align}
    and $\mathcal{A}_K$ has compact resolvent and $0\not\in\sigma\left(\mathcal{A}_K\right)$.
\end{theorem}
\begin{proof}
Denote the corresponding measure corresponding to the metric $g$ as $\dd \nu_g=f\dd x \dd y$ where $f=\sqrt{|\det g|}$ and assume that $g$ is diagonal. Remember
\begin{equation}
    |\nabla \phi|^2_g:=g^{ij} \overline{\nabla_i \phi} \nabla_j \phi.
\end{equation}
We integrate with respect to $\dd\nu_g$ in all the integrals. 
In order to obtain a self-adjoint operator associated to $a(u,v)$, estimate
\begin{align}
    \label{eq:coercivity-aform}
    \begin{split}
        \left|a(u,T_1 u)\right|&=\left|\epsilon_+ \int_{\Omega_+} \overline{\nabla u_+}\nabla u_+ - \epsilon_-\int_{\Omega_-} \overline{\nabla u_-}\nabla\left(-u_- +2 \xi P u_+\right)\right|\\
       & = \left|\epsilon_+||\nabla u_+||^2 + \epsilon_- ||\nabla u_-||^2 - 2\epsilon_-\int_{\Omega_-} \overline{\nabla u_-} \left(\xi \nabla Pu_+ + Pu_+\nabla\xi \right)\right| \\
       &\geq \epsilon_-\left[ ||\nabla u_+||^2 \left(\frac{\epsilon_+}{\epsilon_-} - \frac{1}{\delta}\right) + ||\nabla u_-||^2 \left(1-\delta-\eta\right) - ||u_+||^2 \frac{||\nabla\xi||_{L^\infty(\Omega_+)}}{\eta} \right]
    \end{split}
\end{align}
where we used Young inequality for $\delta$, $\eta>0$ and reversed triangle inequality $|x-y|\geq|x|-|y|$. Second equation was estimated using Cauchy-Schwarz inequality, integral substitution and properties of $\xi$:
\begin{align}
    \begin{split}
        \label{eq:estimate-mirroring-norm}
        \left|\int_{\Omega_-} \overline{\nabla u_-}\left(Pu_+\nabla\xi\right)\right| &\leq ||\nabla u_-||\, ||(Pu_+ \nabla\xi )_-|| \leq ||\nabla u_-||\, ||u_+||\, ||\nabla\xi||^2_{L^\infty(\Omega_+, g)},\\
        \left|\int_{\Omega_-} \overline{\nabla u_-}\left(\xi\nabla Pu_+\right)\right| &\leq ||\nabla u_-||\, ||(\xi\nabla Pu_+)_-|| \leq ||\nabla u_-||\, ||\nabla u_+||.
    \end{split}
\end{align}
In the last estimate~\eqref{eq:estimate-mirroring-norm}, we have used, for $ u=( u_+, u_-)\in L^2(\Omega_+,g)\oplus L^2(\Omega_-,g)$,
 \begin{align}
    \begin{split}
        ||(\chi \nabla& P  u_\pm)_\mp||^2_g = \int_{\Omega_\mp} \abs{\chi \nabla P  u_\pm}^2 \dd \nu_g=\int_{\Omega_\mp} \abs{\chi}^2 \left(g^{ij} \nabla_i P  u_\pm \nabla_j P u_\pm\right) f \dd x \dd y\\
        &=\int_{\Omega_\pm} \abs{P\chi}^2 \left(\left(Pg^{ij}\right) \nabla_i  u_\pm \nabla_j  u_\pm\right) Pf \dd x \dd y = \int_{\Omega_\pm} \abs{\chi}^2 \left(g^{ij} \nabla_i  u_\pm \nabla_j  u_\pm\right) f \dd x \dd y\\
        &=||(\chi \nabla u)_\pm||^2_g\leq ||\nabla u_\pm||^2_g.
    \end{split}
\end{align}

In order to compensate for the last negative term without derivatives of $u$ in~\eqref{eq:coercivity-aform}, define a complexified form $b_t:H^1_0(\Omega_0,g)\times H^1_0(\Omega_0,g)\to L^2(\Omega_0,g)$ for $t\in\mathbb{R}$, $t>0$ as
\begin{equation}
    b_t(u,v):=a(u,v)+\iu t (u,v).
\end{equation}
for $u,v\in H^1_0(\Omega_0,g)$. This sesquilinear form satisfies 
\begin{align}
    \begin{split}
        |b_t(u,u)|&\geq t ||u||^2,\\
        |b_t(u,v)|&\geq |a(u,v)| - t||u||\, ||v||\\
    \end{split}
\end{align}
and is bounded in $H^1_0(\Omega_0,g)$ norm with constant $C_t>0$ due to Poincaré inequality according to
 \begin{align}
    \begin{split}
        |b_t(u,v)|&\leq |a(u,v)| + t ||u||\, ||v||\\
                  &\leq \max\{\epsilon_+,\epsilon_-\}||\nabla u||\, ||\nabla v|| + t||u||\, ||v||\\
                  &\leq C_t ||u||_{H^1_0(\Omega_0,g)} ||v||_{H^1_0(\Omega_0,g)}.
    \end{split}
\end{align}

Combining estimates on $|a(u,T_1 u)|$ with boundedness of $T_1$
\begin{align}
    \begin{split}
        ||u||\, ||T_1 u||&= ||u|| \sqrt{||u_+||^2 + \int_{\Omega_-}\left|-u_-+2\xi Pu_+\right|^2}\\
        &\leq ||u||\, \sqrt{||u_+||^2 + 2(||u_-||^2 + 4 ||u_+||^2)} \leq ||u||\, 3||u|| = 3||u||^2,
    \end{split}
\end{align}
we obtain for $\beta\in\mathbb{R}$, $\beta>0$
 \begin{align}
     \begin{split}
         |b_t(u,u)|+|b_t(u,\beta T_1 u)| &= |b_t(u,u)|+\beta |b_t(u,T_1 u)|\\
                 &\geq ||u_+||^2\left(t(1-3\beta) - C_\eta\beta\right) + t||u_-||^2\left(1-3\beta\right)\\
                 &\quad + \beta\epsilon_-||\nabla u_+||^2\left(\frac{\epsilon_+}{\epsilon_-}-\frac{1}{\delta}\right) + \beta\epsilon_-||\nabla u_-||^2 \left(1-\delta-\eta\right)\\
                 &\geq \alpha ||u||_{H_0^1(\Omega_0,g)}
     \end{split}
 \end{align}
where $C_\eta=\epsilon_- ||\nabla\xi||^2_{L^\infty(\Omega_+)}\eta^{-1}$. For a choice of $0<\beta<\frac{1}{3}$, $0<\delta<\frac{\epsilon_-}{\epsilon_+}<1$, there exists $\eta>0$ and $t>t_0$ sufficiently large such that each term on the right-hand side is strictly positive. From here, it is trivial to provide lower bound in terms of $H_0^1(\Omega_0,g)$ norm such that $\alpha>0$ is strictly positive for $\frac{\epsilon_+}{\epsilon_-}>1$. By the same computation with $b(u,T_2u)$, we obtain similar strictly positive bounds for $\frac{\epsilon_+}{\epsilon_-}<1$.

For $\frac{\epsilon_+}{\epsilon_-}>1$, set $\mathcal{V}=H_0^1(\Omega_0,g)$, $\mathcal{H}=L^2(\Omega_0,g)$ and $\Phi_1=\beta T_1$ in Theorem~\ref{thm:almog2}. Form $a$ is symmetric and thus the second inequality in~\eqref{eq:almog-twocond} is also satisfied for $\Phi_2=\beta T_1$. This implies that operator $\mathcal{B}_t$ associated to form $b_t$ defined on $S\subset L^2(\Omega_0,g)$ is a closed isomorphism from a dense subset of $H^1_0(\Omega_0,g)$ to $L^2(\Omega_0,g)$. At the same time, $\mathcal{B}_t$ has a bounded inverse and due to compact embedding\footnote{This result is known as Rellich-Kondrakov theorem~\cite{hebey2000nonlinear}.} of $H^1_0(\Omega_0,g)$ to $L^2(\Omega_0,g)$, its resolvent $R(\lambda,\mathcal{B}_t)$ is compact for $\lambda=0$ and by the first resolvent identity also for all $\lambda$ in the resolvent set. Hence, the essential spectrum is empty~\cite[Theorem 6.29]{kato2013perturbation}. By setting $\Phi_{1,2}=\beta T_2$, we obtain the same results also for $\frac{\epsilon_+}{\epsilon_-}<1$.

To determine domain of the operator --- by Riesz theorem, stating that every continuous linear functional $u\mapsto \varphi(u)$ on $L^2(\Omega_0,g)$ is represented by some $\eta\in L^2(\Omega_0,g)$ such that $\varphi(u)=(\eta,u)$, we have
 \begin{equation}
     \dom \mathcal{B}_t=\left\{ v\in H_0^1(\Omega_0,g): \exists\eta\in L^2(\Omega_0,g)\,\forall u\in H_0^1(\Omega_0,g), b_t(u,v)=(u,\eta)\right\}.
\end{equation}

Based on definition of weak derivatives, it follows that for $u\in C_0^\infty(\Omega_0)$ and $v$ from 
\begin{align}
    \begin{split}\label{eq:dom-bt-rigorous}
        \dom \mathcal{B}_t &= \left\{v\in H_0^1(\Omega_0,g): \nabla\cdot(\epsilon\nabla v)\in L^2(\Omega_0,g)\right\}\\
    \end{split}
\end{align}
we obtain
\begin{align}\label{eq:bt-eta-def-weak-der}
    \begin{split}
        b_t(u,v)&=a(u,v)+\iu t (u,v)=\int_{\Omega_0}\overline{\nabla u}\epsilon\nabla v  + \iu t (u,v)\\
            &=-\int_{\Omega_0} \overline{u}\nabla\cdot\left(\epsilon \nabla v\right) + \iu t (u,v)=(u,\eta)
    \end{split}
\end{align}
where we have used the definition of weak derivative of $\epsilon \nabla v$ and piece-wise constant $\epsilon$. From density of $C_0^\infty(\Omega_0)$ in $H_0^1(\Omega_0,g)$, $b_t(u,v)=(u,\eta)$ holds also for $u\in H_0^1(\Omega_0,g)$. 

Introduce weak-sense notation 
\begin{equation}
    \dom \mathcal{B}_t=\left\{v\in H_0^1(\Omega_0,g): \Delta v_\pm \in L^2(\Omega_\pm),\ \restr{\left(\epsilon_+\partial_x v_+ +\epsilon_-\partial_x v_-\right)}{\mathcal{C}} =0\right\},
\end{equation}
bearing exactly the same meaning as in~\eqref{eq:dom-bt-rigorous} but with emphasis given on the fact that the interface condition is present in some weak-sense.  To illustrate, consider $v=(v_+,v_-)\in H_0^1(\Omega_0)$, $v_\pm \in H^2(\Omega_\pm)$. Then we could apply the divergence theorem on $\Omega_\pm$ and obtain for all $u\in H_0^1(\Omega)$
\begin{equation}
    (\nabla u,\epsilon\nabla v)=-(u,\ddiv \left(\epsilon\nabla v\right)) + \int_{\mathcal{C}} \overline{u} \left(\epsilon_+ \partial_x v_+ + \epsilon_- \partial_x v_-\right) \dd y.
\end{equation}
At the same time we have, from the definition of weak derivative, equation~\eqref{eq:bt-eta-def-weak-der} and in conclusion, $\epsilon_+ \partial_x v_+ + \epsilon_- \partial_x v_- = 0$ on $\mathcal{C}$ in the sense of traces. Regularized normal traces in rigorous setting were introduced by Behrndt and Krejčiřík in~\cite{behrndt2014indefinite} also for functions $f$ such that $f\in H_0^1(\Omega_0) \ \land \ \ddiv\left(\epsilon\nabla f\right)\in L^2(\Omega_0)$ and provide details on the matter.

To extract information about operator without a complex shift, define for $t\in\mathbb{R}$, $t>0$,
 \begin{equation}\label{eq:form-operator}
     \mathcal{A}_K:=\mathcal{B}_t-\iu t I.
 \end{equation}
From~\eqref{eq:form-operator} we have
\begin{equation}
    \mathcal{A}_K^*:=\mathcal{B}_t^*+\iu t I
\end{equation}
and also $\dom\mathcal{A}_K^*=\dom \mathcal{B}_t^* = \dom \mathcal{B}_t=\dom\mathcal{A}_K$ due to operator $\mathcal{B}_t^*$ being associated to $b_t^*(u,v)=\overline{b_t(v,u)}=a(u,v)-\iu t(u,v)$ from symmetricity of form $a$. At the same time, for $u\in H_0^1(\Omega_0,g)$, $v\in \dom \mathcal{A}_K^*$,
\begin{align}
    \begin{split}
        (u, \mathcal{A}_K^* v)&=(u, \left(\mathcal{B}_t^*+\iu t\right)v)=b_t^*(u,v)+\iu t(u,v) = \overline{a(v,u)}-\iu t(u,v)+\iu t(u,v)\\
    &=a(u,v)=(u,\mathcal{A}_K v)
    \end{split}
\end{align}
and hence, $\mathcal{A}_K^*=\mathcal{A}_K$ is self-adjoint and independent of $t$.
\end{proof}
\begin{remark}\label{remark:generalization-forms}
    Result of Theorem~\ref{theorem:form-A-self-adjoint} can be extended to Riemannian manifolds $(\mathcal{N},\dot g)$, $\dot g_{ij}(x, y) = g_{ij}(-x,y)$, or in terms of Gaussian curvature $K$: $K(x,y)=K(-x,y)$, in some rectangular neighbourbood of $\Gamma$, as illustrated in Figure~\ref{fig:form-generalization}.
\end{remark}

Finally, we would like to show that the operator $\mathcal{A}_K$ coincides with $A_K$ for non-critical contrast $\kappa\not = 1$.
\begin{proposition}\label{prop:form-operator-same}
    The operator $\mathcal{A}_K$ defined in Theorem~\ref{theorem:form-A-self-adjoint} via forms coincides, for $\kappa\not=1$, with the self-adjoint operator $A_K\supset \dot A_K$ defined in Section~\ref{section:essential-self-adjoint}.
\end{proposition}
\begin{proof}
    First, we will prove that $\dot A_K\subset \mathcal{A}_K$. As 
    \begin{equation}
            \dom \dot A_K= \left\{
            \begin{array}{l|l}
             &\psi_\pm\mid_{\pp \Omega_0}=0, \\
            \psi=\begin{pmatrix}\psi_+\\\psi_-\end{pmatrix} \in H^2(\Omega_+,g)\oplus H^2(\Omega_-,g) &\psi_+(0,\cdot) = \psi_-(0,\cdot)\\
            &\epsilon_+ \pp_1\psi_+(0,\cdot) = -\epsilon_- \pp_1\psi'_-(0,\cdot)
            \end{array}
        \right\},
    \end{equation}
    we have that $\dom \dot A_K\subset\dom \mathcal{A}_K$. Based on the action of the operators, we have
    \begin{equation}
        (u,\dot A_K v)=a(u,v)=(u,\mathcal{A}_K v),\quad \forall u\in C_0^\infty(\Omega_0),\ v\in\dom\dot A_K
    \end{equation}
    and by density of test functions then $\dot A_K\subset \mathcal{A}_K$.
    Whenever there are two densely-defined symmetric operators $L_1$, $L_2$ on a Hilbert space, then the following holds for their adjoints $L_1^*$,  $L_2^*$,
    \begin{equation}
        L_1\subset L_2\implies L_2^*\subset L_1^*.
    \end{equation}
    Hence, $A\subset \mathcal{A}_K$ as $\mathcal{A}_K^*\subset \dot A_K^*$ and
    \begin{equation}
        A_K=\overline{\dot A_K}=\dot A_K^{**}\subset \mathcal{A}_K^{**}=\mathcal{A}_K.
    \end{equation}
    And finally, we obtain also the second extension
    \begin{equation}
        \mathcal{A}_K=\mathcal{A}_K^*\subset A^*=A_K.
    \end{equation}
\end{proof}
\begin{note}[Cut-off motivation]
    We can see that if we had used no cut-off in~\eqref{eq:T-transform}, i.e. $\xi\equiv 1$ constant, the last term in~\eqref{eq:coercivity-aform} with $||u_+||$ would be zero and $\mathbb{T}$-coercivity of form $a$ would be achieved rather quickly. Although first, we have to properly define mirroring operator $P$ globally via zero extension. It turns out that with this simpler choice, we cannot recover proper definition of $\mathcal{A}_K$ via forms for the whole range of contrasts $\kappa\in\mathbb{R}$, $\kappa\not =1$.

Denote $a_*:=\min\{a,b\}$. For $T_1$ and $T_2$, use $R_\pm=P_\pm$, $P_+: H^1_{0,\mathcal{C}}(\Omega_+)\to H^1_{0,\mathcal{C}}(\Omega_-)$ and $P_-: H^1_{0,\mathcal{C}}(\Omega_-)\to H^1_{0,\mathcal{C}}(\Omega_+)$, respectively, defined by
    \begin{equation}\label{eq:mirroring}
        (P_\pm u_\pm)(x,y):=
        \begin{cases}
            u_\pm(-x,y) & \mathrm{for}\ x\in(-a_*,a_*), \\
            0 &\mathrm{for}\ x\not\in(-a_*,a_*).
        \end{cases}
    \end{equation}
    In order for $P_\pm$ to be also an operator $H^1_{0,\mathcal{C}}(\Omega_\pm)\to H^1_{0,\mathcal{C}}(\Omega_\mp)$, we require that $b\geq a$, or $a\geq b$, respectively (plus and minus). This can be show to lead to the following situation:
    \begin{align}
        \begin{split}
            b \geq a \implies \mathrm{form}\ a(\cdot,T_1 \cdot)\ \mathrm{is\ coercive\ for}\ \kappa > 1,\\
            a \geq b \implies \mathrm{form}\ a(\cdot,T_2 \cdot)\ \mathrm{is\ coercive\ for}\ \kappa < 1.
        \end{split}
    \end{align}
    As can be seen, it is impossible to recover the condition for both $\kappa>1$ and $\kappa<1$ unless $a=b$.

    Another approach one can employ is scaling $\Omega_\pm$ such that, under this transformation, they have same dimensions. For simplicity, assume $b\geq a$. By using $T_1$ as above, this gives us coercivity for $\kappa >1$. By defining $T_2$ using different choice of operator $R_-=P$ as 
    \begin{align}
        \begin{split}
            (Pu_-)(x,y)=u_-\left(-\frac{b}{a}x,y\right),
        \end{split}
    \end{align}
    we arrive at $|a(u,u)|+|a(u,T_2u)|\geq \alpha ||u||_{H^1(\Omega_0,g)}$ for $\alpha>0$, all $u\in H_0^1(\Omega_0,g)$ and for $\kappa < \frac{a}{b}$. Overall, combining reflections in $T_2$ and reflections with rescaling in $T_1$, we obtain self-adjoins representations with compact resolvent via form $a(u,v)$ for $\kappa\in(0,\frac{a}{b})\cup(1,\infty)$. Thus, we see the neccessity of including the cut-off.
\end{note}

\section{Critical contrast}\label{section:critical-contrast}
\subsection{Singular sequences for zero}
In this section, we will show that for the critical contrast $\frac{\epsilon_+}{\epsilon_-}=:\kappa=1$, the operator $A_K$ defined on a manifold with arbitrary constant curvature always contains zero in the essential spectrum.

\begin{theorem}\label{theorem:singular-seq}
For critical contrast $\kappa$, there is zero in the essential spectum for arbitrary Gaussian curvature $K\in\mathbb{R}$, i.e. $\kappa=1 \implies$ $0\in\sigma_\mathrm{ess}(A_K)$.
\end{theorem}
\begin{proof}
    We can easily see, that the function $(x,y)\mapsto\sin\left(\frac{n\pi}{c}y\right)$ satisfies Dirichlet boundary conditions on boundaries perpendicular to the $y$ axis. For this reason, we will now try to find suitable singular sequences in the form $(x,y)\mapsto g(x)\sin\left(\frac{n\pi}{c}y\right)$, basically making a separation of variables.

    To construct our singular sequences, we will utilize equation
    \begin{equation}
        A_K\psi(x,y)=0.
    \end{equation}

    Consider ansatz $\psi(x,y)=f(x)\sin\left(\frac{n\pi}{c}y\right)$. Then the solution for $f$ of
    
    \begin{equation}\label{eq:singular-pseudomode}
        A_K\psi(x,y)=\left(-f''(x) + \sqrt{K}\tan(\sqrt{K}x)f'(x) + \frac{(\frac{n\pi}{c})^2}{\cos^2(\sqrt{K}x)}f(x)\right)\sin\left(\frac{n\pi}{c}y\right) = 0
    \end{equation}
    is in a general form of
    \begin{equation}
        f(x) = C_1 \cosh\left( \frac{n\pi}{c \sqrt{K}}\arctanh{\sin\left(\sqrt{K}x\right)}\right) + C_2 \sinh\left( \frac{n\pi}{c \sqrt{K}}\arctanh{\sin\left(\sqrt{K}x\right)}\right),
    \end{equation}
    for all $K\in \mathbb{R}\setminus\{0\}$ where $C_1$, $C_2 \in \mathbb{C}$ are constants. We are not going to construct eigenvectors of $A_K$ corresponding to eigenvalue $\lambda=0$ as it is not an eigenvalue in general (only for $a=b$). Instead, we will construct approximations. From these solutions, by a choice $C_1=1$ and $C_2=-1$, we define functions $f_n^{(K)}: (0, a)\to \mathbb{R}$ using
    \begin{equation}\label{eq:fnK}
        f_n^{(K)}(x):=\exp\left(-\frac{n\pi}{c \sqrt{K}}\arctanh{\sin(\sqrt{K}x)}\right)
    \end{equation}

    for each $n\in \mathbb{N}$ and fixed $K\in \mathbb{R}\setminus\{0\}$. For case $K=0$ it is defined via $f_n^{(0)}(x):=\exp\left(-\frac{n\pi}{c}x\right)$. But this case is already present in a limit sense $\lim_{K\to0}f_n^{(K)}(x)=f_n^{(0)}(x)$, as can be seen from Taylor expansion in $K$. When the curvature $K$ is obvious from context, we will \emph{denote the function only by} $f_n$. These functions with slight modifications will be our candidates for a singular sequence for zero. Mainly, we need to modify them to satisfy the interface and Dirichlet boundary conditions.

    \begin{figure}
        \centering
        \includegraphics[width=1\linewidth]{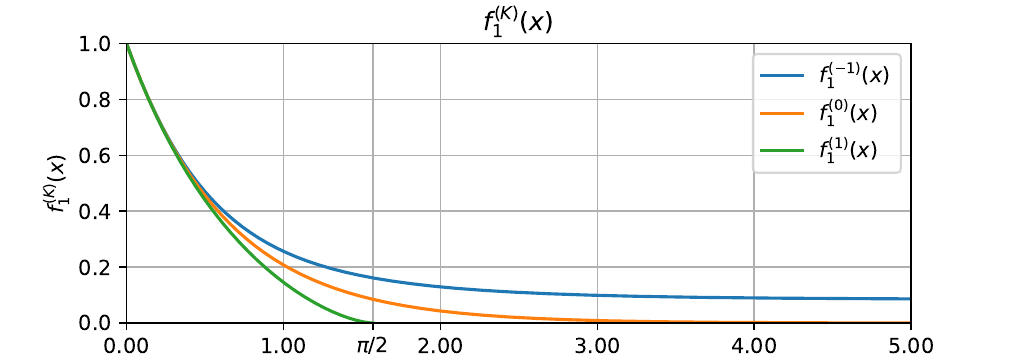}
        \caption{Functions $f_n^{(K)}$ for $n=1$ given in~\eqref{eq:fnK}. We have  $f_n^{(K)}(x)=(f_1^{(K)}(x))^{\frac{n\pi}{c}}$ For positive curvature, we have $a,b\in(0,\frac{\pi}{2})$ so we plot only those valid values. For negative curvature, the function does have a non-zero limit as $x\to\infty$.}
    \end{figure}

    To that end, we will assume $a\leq b$ (for $a=b$, $\lambda=0$ is an infinitely degenerate eigenvalue) and fix two parameters $a_1$, $a_2 \in \mathbb{R}$, $0<a_1<a_2<a$, $a_2<b$ and define a smooth cut-off function $\chi: (0,a) \to \mathbb{R}$ with properties
    \begin{equation}
        \chi(x):=
        \begin{cases}
            1, & x\in(0, a_1), \\
            0, & (a_2, a)
        \end{cases}
    \end{equation}
    and with values elsewhere given such that $\chi\in C^\infty((0,a))$. Define sequence $\{\vphi_n\}_{n\in \mathbb{N}}\subset\dom A_K$ for operator $A_K$ as
    \begin{equation}
        \vphi_n(x,y):= 
        \begin{cases}
            f_n(|x|)\chi(|x|)\sqrt{\frac{2}{c}}\sin\left(n\frac{\pi}{c}y\right), & x\in(-a,a),\\
            0, & x\in(-b,-a).\\
        \end{cases}
    \end{equation}
    Indeed, it clearly satisfies Dirichlet boundary conditions and also the interface conditions.
%
    We note that $\lambda=0$ is not an eigenvalue for $a\not =b$, it does not belong to discrete spectrum $\sigma\setminus\sigma_\mathrm{ess}$, and sequences $(\vphi_n)_{n\in\mathbb{N}}$ for all three cases are therefore singular.
    
    \paragraph{Case $K=0:$} The functions $f_n$ reduce to $f_n(x)=\exp\left(-\frac{n\pi}{c}|x|\right)$. We will estimate
    \begin{equation}
        \|\vphi_n\|^2=2\int_0^a \left(f_n\chi\right)^2=2\int_0^{a_2} \left(f_n\chi\right)^2 \ge 2\int_0^{a_1} \exp\left(-2\frac{n\pi}{c}x\right) \dd x = \frac{c}{n\pi}\left(1-\e^{-2\frac{n\pi}{c}a_1}\right)
    \end{equation}
    and evaluate the following expression for $\chi\in C^\infty((0,a))$:
    \begin{align}
        \begin{split}
            \|A \vphi_n\|^2&=\int_{-b}^{a} \e^{-\frac{2n\pi}{c}|x|}\left(-\frac{2n\pi}{c}\sgn(x)\chi'(|x|)+\chi''(|x|)\right)^2 \dd x \le 2C_n  \int_{a_1}^{a_2} \e^{-\frac{2n\pi}{c}x} \dd x \\
                 &= \frac{C_nc}{n\pi}\left(\e^{-\frac{2n\pi}{c}a_1}-\e^{-\frac{2n\pi}{c}a_2}\right) = \frac{C_nc}{n\pi} \e^{-\frac{2n\pi}{c}a_1}\left(1-\e^{-(a_2-a_1)\frac{2n\pi}{c}}\right), 
        \end{split}
    \end{align}
    where we estimated the first bracket using a degree two polynomial $C_n$ in $n$. Hence, $\frac{\vphi_n}{\|\vphi_n\|}$ is a singular sequence as the expression $\lim_{n\to\infty}\frac{1}{\|\vphi_n\|}\|A\vphi_n\|=0$ is zero.

    \paragraph{Case $K=+1:$} The functions $f_n$ reduce to $f_n(x)=\left(\frac{1+\sin |x|}{1-\sin |x|}\right)^{-\frac{n\pi}{2c}}$ using a known relation for $\arctanh(y)=\frac{1}{2}\ln(\frac{1+y}{1-y})$. Additionally, we have $f_n'(x)=-\frac{n\pi}{c}\frac{\sgn x}{\cos x}f_n(x)$. The following estimates are valid because of convexity of $\left(\frac{1+\sin x}{1-\sin x}\right)^{-\frac{1}{2}}$, as the second derivative is non-negative on $[0,\frac{\pi}{2}]$:
    \begin{equation}
        1-x\le \left(\frac{1+\sin x}{1-\sin x}\right)^{-\frac{1}{2}} \le 1-\frac{2}{\pi}x, \quad x\in\left(0,\frac{\pi}{2}\right).
    \end{equation}
    The lower bound is a tangent at 0 to the convex function and the upper bound is a secant to the function crossing points $x=0$ and $x=\frac{\pi}{2}$. Choose $a_1<a_2<1$. Then we estimate
    \begin{equation}
        \|\vphi_n\|^2 = 2\int_0^a \left(f_n\chi\right)^2 \dd \mu_{+1} \ge 2\cos(a_2)\int_0^{a_1} \left(1-x\right)^{\frac{2n\pi}{c}} \dd x = \frac{\cos(a_2)}{\frac{2n\pi}{c}+1}\left(1 - \left(1-a_1\right)^{\frac{2n\pi}{c}+1}\right),
    \end{equation}
    where $\dd\mu_{+1}=\cos(x)\dd x$ is measure on the rectangle with curvature $K=+1$. Continue to give an upper bound for our expression of interest:
    \begin{align}
        \begin{split}
            \|A \vphi_n\|^2&=\int_{-b}^{a} f_n^2(x)\left(-\frac{2n\pi}{c}\frac{\sgn x}{\cos x}\chi'(|x|)+\chi''(|x|) - \tan\left(x\right)\chi'(|x|)\right)^2 \dd \mu_{+1}\\
         &\le 2C_n  \int_{a_1}^{a_2} \left(1-\frac{2}{\pi}x\right)^{\frac{2n\pi}{c}} \dd x = \frac{\pi C_n}{\frac{2n\pi}{c}+1}\left(\left(1-\frac{2}{\pi}a_1\right)^{\frac{2n\pi}{c}+1} - \left(1-\frac{2}{\pi}a_2\right)^{\frac{2n\pi}{c}+1}\right) \\
         &= \frac{\pi C_n}{\frac{2n\pi}{c}+1} \left(1-\frac{2}{\pi}a_1\right)^{\frac{2n\pi}{c}+1}\left(1 - \left(\frac{1-\frac{2}{\pi}a_2}{1-\frac{2}{\pi}a_1}\right)^{\frac{2n\pi}{c}+1}\right),
        \end{split}
    \end{align}
    where we employed boundedness of $\frac{1}{\cos(x)}$ a $\tan(x)$ on $(0,a)\subsetneq(0, \frac{\pi}{2})$ and $C_n$ is again a polynomial in $n$. Same as before, $\frac{\vphi_n}{\|\vphi_n\|}$ is a singular sequence because the expression $\lim_{n\to\infty}\frac{1}{\|\vphi_n\|}\|A\vphi_n\|=0$ is zero.

    \paragraph{Case $K=-1$:} \qquad Now we have $f_n(|x|)=\exp(-\frac{n\pi}{c}\arctan\sinh |x|)$ and $f_n'(|x|)=-\frac{n\pi}{c}\frac{\sgn x}{\cosh(x)}f_n(|x|)$. The following estimates are, again, valid because of convexity and the fact that $\lim_{x\to\infty}\e^{-\arctan\sinh x}=\e^{-\frac{\pi}{2}}<\frac{1}{2}$:
    \begin{equation}
        1-x\le \exp(-\arctan\sinh x) \le \tilde f(x), \quad x\in(0,\infty),
    \end{equation}
    where function $\tilde f: (0,\infty)\to\mathbb{R}$ is given by
    \begin{equation}
        \tilde f(x) :=
        \begin{cases}
            1-\frac{x}{2}, &x\in(0,1),\\
            \frac{1}{2}, &x\in(1,\infty).
        \end{cases}
    \end{equation}
    
    Let us choose cut-off constants $a_1$ and $a_2$ such that $0<a_1<a_2<1$ and as a consequence $\forall n\in\mathbb{N}, \forall x\in(a_1,a_2):f_n(x)\leq\left(1-\frac{x}{2}\right)^{\frac{n\pi}{c}}$. Even with measure $\dd\mu_{-1}=\cosh(x)\dd x$, the expression $\|B_k\vphi_n\|$ will behave very similarly as before because we have already examined a similar upper bound on $f_n$ and also functions $\frac{1}{\cosh x}$ and $\tanh x$ are bounded on $(0, a) \subsetneq (0,\infty)$. Concluding, we have found a singular sequence for $K=-1$.
\end{proof}
\begin{remark}
    It is possible that a similar construction can be given also for non-constant curvatures with metric $g(x,y)=\mathrm{diag}(1,f(x))$. The ordinary differential equation
    \begin{equation}
        A_K \left[\psi(x)\sin\left(\frac{m\pi}{c}y\right)\right]=0
    \end{equation}
    has solutions given as
    \begin{equation}
        \psi(x):=C_1\exp\left(-\frac{m\pi}{c}\int\frac{1}{f(x)}\dd x\right) + C_2\exp\left(\frac{m\pi}{c}\int\frac{1}{f(x)}\dd x\right).
    \end{equation}
    for arbitrary constants $C_1$,  $C_2$ as can be found by reducing the problem to a system of first order ODEs. For a choice of  $C_1=1$, $C_2=0$, we have
     \begin{equation}
        \psi(x)=\exp\left(-\frac{m\pi}{c}\int\frac{1}{f(x)}\dd x\right).
    \end{equation}
    We have not explored if these functions (for arbitrary $f(x)$) lead to singular sequences due to time concerns. Also, due to usage of cut-off functions, the geometry of the domain could be much richer and the results of this section would still apply, similarly to Remark~\ref{remark:generalization-forms}.
\end{remark}

\subsection{Improved characterisation for flat manifolds}
In this section, we will give a full proof of characterisation of essential spectrum $\sess(A_0)$ depending on values $\epsilon_+$ and $\epsilon_-$ for zero Gaussian curvature $K=0$. Remember that eigenvalues are given by Proposition~\ref{prop:2d-spectrum}.
In the interval $\lambda \in \left(-\epsilon_-(\frac{n\pi}{c})^2,\epsilon_+(\frac{n\pi}{c})^2\right)$, the characteristic equation becomes
\begin{equation}
    \frac{\tanh(a\sqrt{(\frac{n\pi}{c})^2-\frac{\lambda}{\epsilon_+}})}{\epsilon_+\sqrt{(\frac{n\pi}{c})^2-\frac{\lambda}{\epsilon_+}}} = \frac{\tanh(b\sqrt{\frac{\lambda}{\epsilon_-} + (\frac{n\pi}{c})^2 }) }{\epsilon_- \sqrt{\frac{\lambda}{\epsilon_-} + (\frac{n\pi}{c})^2}}.
    \label{eq:K0-eigenequation}
\end{equation}


\begin{lemma}
    Let $\epsilon_- = \epsilon_+ =: \epsilon > 0$ and choose a fixed $n\in \mathbb{N}$. Then the equation~\eqref{eq:K0-eigenequation} has exactly one solution $\lambda$ in interval $\left(-\epsilon\left(\frac{n\pi}{c}\right)^2, \epsilon\left(\frac{n\pi}{c}\right)^2\right)$. If $a=b,\ b<a,\ b>a$, then $\lambda=0,\ \lambda<0,\ \lambda>0$, respectively.
        \label{lemma:B0-root-uniqueness}
\end{lemma}

\begin{proof}
For each $n\in\mathbb{N}$ define a function $G_n: \left(-\epsilon_-(\frac{n\pi}{c})^2,\ \epsilon_+(\frac{n\pi}{c})^2 \right) \to\mathbb{R}$ as a difference of reciprocals of left and right-hand sides of equation~\eqref{eq:K0-eigenequation}:

\begin{equation}
    G_n(\lambda):=\frac{\epsilon \sqrt{\frac{\lambda}{\epsilon} + (\frac{n\pi}{c})^2}}{ \tanh(b\sqrt{\frac{\lambda}{\epsilon} + (\frac{n\pi}{c})^2 }) }
 -
    \frac{\epsilon\sqrt{(\frac{n\pi}{c})^2-\frac{\lambda}{\epsilon}}}{ \tanh(a\sqrt{(\frac{n\pi}{c})^2-\frac{\lambda}{\epsilon} })}.
\end{equation}

After rearranging derivative $G_n'$ into (similar rearrangements as in article~\cite{behrndt2014indefinite}):
\begin{equation}
    G_n'(\lambda) = \frac{\sinh\left(2a\sqrt{(\frac{n\pi}{c})^2 - \frac{\lambda}{\epsilon}}\right) - 2a\sqrt{(\frac{n\pi}{c})^2 - \frac{\lambda}{\epsilon}}}{4\sinh\left(a\sqrt{(\frac{n\pi}{c})^2 - \frac{\lambda}{\epsilon}}\right)^2 \sqrt{(\frac{n\pi}{c})^2 - \frac{\lambda}{\epsilon}}}
    + \frac{\sinh\left(2b\sqrt{(\frac{n\pi}{c})^2 + \frac{\lambda}{\epsilon}}\right) - 2b\sqrt{(\frac{n\pi}{c})^2 + \frac{\lambda}{\epsilon}}}{4\sinh\left(b\sqrt{(\frac{n\pi}{c})^2 - \frac{\lambda}{\epsilon}}\right)^2 \sqrt{(\frac{n\pi}{c})^2 + \frac{\lambda}{\epsilon}}},
\end{equation}
we readily obtain, using an identity $\sinh x > x$ valid for all $x>0$, statement $G_n'(\lambda)>0$ valid on the whole domain of $G_n$.

Limit $\lim_{\lambda\to\epsilon_+(\frac{n\pi}{c})^2-}G_n(\lambda)$ is positive and a similar limit on the other end of domain $\lim_{\lambda\to-\epsilon_-(\frac{n\pi}{c})^2+}G_n(\lambda)$ is negative. From this fact, and from continuity of $G_n$, it follows that there exists exactly one root of characteristic equation~\eqref{eq:K0-eigenequation} in the domain of $G_n$. It is also easy to determine the sign of the root --- using the value
\begin{equation}
    G_n(0) = \epsilon\frac{n\pi}{c} \left( \frac{1}{\tanh\left(b\frac{n\pi}{c}\right)} - \frac{1}{\tanh\left(a\frac{n\pi}{c}\right)} \right)
\end{equation}
and a sign of the root is determined by a sign of $G_n(0)$. This is because intersection of graph $G_n$ with axis $\lambda=0$ is exactly one, see above.
\end{proof}

\begin{proposition}
    $\sigma_{\mathrm{ess}}(A_0)=\{0\} \iff \epsilon_+=\epsilon_-$.
    \label{prop:B0-0-spectrum}
\end{proposition}

\begin{proof}

The other implication $\epsilon_+=\epsilon_-\implies\Lambda=0$ is harder to prove. We start by putting $\epsilon:=\epsilon_+=\epsilon_-$, fixing arbitrary $n\in \mathbb{N}$ and denoting $\lambda_n$ the root of the characteristic equation
 \begin{equation}
    \frac{\tanh(a\sqrt{(\frac{n\pi}{c})^2-\frac{\lambda}{\epsilon}})}{\sqrt{(\frac{n\pi}{c})^2-\frac{\lambda}{\epsilon}}} = \frac{\tanh(b\sqrt{\frac{\lambda}{\epsilon} + (\frac{n\pi}{c})^2 }) }{\sqrt{\frac{\lambda}{\epsilon} + (\frac{n\pi}{c})^2}}.
    \label{eq:K0-eigeneq-same}
\end{equation}
lying in $\left(0, \epsilon\left(\frac{n\pi}{c}\right)^2\right)$. According to Lemma~\ref{lemma:B0-root-uniqueness}, such a root exists and is unique. Let us define a sequence
\begin{equation}
    \alpha_n:=\frac{\lambda_n}{\epsilon}\left(\frac{c}{n\pi}\right)^2 \in (0,1)
\end{equation}
for all $n\in\mathbb{N}$. Using this sequence, we will prove $\lambda_n\to 0$.

\paragraph{Step 1}
Rewrite equation~\eqref{eq:K0-eigeneq-same} as
\begin{equation}
    a\frac{\tanh(a\frac{n\pi}{c}\sqrt{1-\alpha_n})}{a\frac{n\pi}{c}\sqrt{1-\alpha_n}} = b\frac{\tanh(b\frac{n\pi}{c}\sqrt{1+\alpha_n})}{b\frac{n\pi}{c}\sqrt{1+\alpha_n}}.
\end{equation}
Since the function $x\mapsto \frac{\tanh(x)}{x}$ converges to 0 as $x\to +\infty$ and the sequence $\{\sqrt{1+\alpha_n}\}_n\in(1,\sqrt{2})$, the right hand side converges to 0 as $n\to+\infty$. In other words,
\begin{equation}
    \lim_{n\to\infty}a\frac{\tanh(a\frac{n\pi}{c}\sqrt{1-\alpha_n})}{a\frac{n\pi}{c}\sqrt{1-\alpha_n}} = 0.
\end{equation}
Then, as a result of $x\mapsto \frac{\tanh(x)}{x}$ being is strictly positive for all positive $x$, the following holds:
\begin{equation}
    \lim_{n\to\infty} n\sqrt{1-\alpha_n}=+\infty.
    \label{eq:K0-step1}
\end{equation}

\paragraph{Step 2}
Let us again rearrange~\eqref{eq:K0-eigeneq-same} as
\begin{equation}
    \frac{\tanh\left(a\frac{n\pi}{c}\sqrt{1-\alpha_n}\right)}{\tanh\left(b\frac{n\pi}{c}\sqrt{1+\alpha_n}\right)} = \frac{\sqrt{1-\alpha_n}}{\sqrt{1+\alpha_n}}.
\end{equation}
From previous step~\eqref{eq:K0-step1}, the left-hand side converges to 1 as $n\to+\infty$. Thus,
\begin{equation}
    \lim_{n\to+\infty}\sqrt{\frac{1-\alpha_n}{1+\alpha_n}} = 1.
\end{equation}
As function  $x \mapsto \frac{1-x}{1+x}$ is strictly less than 1 for all $x\in(0,1)$ and $\lim_{x\to 0}\frac{1-x}{1+x}=1$, it follows that
\begin{equation}
    \lim_{n\to\infty}\alpha_n = 0.
\end{equation}

\paragraph{Step 3}
For the last time, let us rearrange and expand hyperbolic functions from~\eqref{eq:K0-eigeneq-same} into
\begin{align}
    \begin{split}
        1-2\frac{\e^{-2a\frac{n\pi}{c}\sqrt{1-\alpha_n}}}{1+\e^{-2a\frac{n\pi}{c}\sqrt{1-\alpha_n}}} &= \left(1-2\frac{\e^{-2b\frac{n\pi}{c}\sqrt{1+\alpha_n}}}{1+\e^{-2b\frac{n\pi}{c}\sqrt{1+\alpha_n}}}\right)\frac{\sqrt{1-\alpha_n}}{\sqrt{1+\alpha_n}}\\
                                                                                                     &= \left(1-2\frac{\e^{-2b\frac{n\pi}{c}\sqrt{1+\alpha_n}}}{1+\e^{-2b\frac{n\pi}{c}\sqrt{1+\alpha_n}}}\right)\left(1-\alpha_n+\mathcal{O}(\alpha_n^2)\right).
    \end{split}
\end{align}
Expanding the right-hand side and rearranging terms,
\begin{equation}
    -2\frac{\e^{-2a\frac{n\pi}{c}\sqrt{1-\alpha_n}}}{1+\e^{-2a\frac{n\pi}{c}\sqrt{1-\alpha_n}}} = -\alpha_n -2\frac{\e^{-2b\frac{n\pi}{c}\sqrt{1+\alpha_n}}}{1+\e^{-2b\frac{n\pi}{c}\sqrt{1+\alpha_n}}} + o(\alpha_n).
\end{equation}
Now, multiplying the equation by a factor $\frac{1}{2\alpha_n}$, bringing the exponentials to the right-hand side and applying $\lim_{n\to\infty}$ to both sides, we obtain
\begin{align}
    \begin{split}
        \frac{1}{2} &= \lim_{n\to\infty} \frac{1}{\alpha_n} \left(\frac{\e^{-2a\frac{n\pi}{c}\sqrt{1-\alpha_n}}}{1+\e^{-2a\frac{n\pi}{c}\sqrt{1-\alpha_n}}} -\frac{\e^{-2b\frac{n\pi}{c}\sqrt{1+\alpha_n}}}{1+\e^{-2b\frac{n\pi}{c}\sqrt{1+\alpha_n}}}\right)\\
                    &= \lim_{n\to\infty}\frac{\e^{-2a\frac{n\pi}{c}\sqrt{1-\alpha_n}} - \e^{-2b\frac{n\pi}{c}\sqrt{1+\alpha_n}}}{\alpha_n} \frac{\frac{\e^{-2a\frac{n\pi}{c}\sqrt{1-\alpha_n}}}{1+\e^{-2a\frac{n\pi}{c}\sqrt{1-\alpha_n}}} -\frac{\e^{-2b\frac{n\pi}{c}\sqrt{1+\alpha_n}}}{1+\e^{-2b\frac{n\pi}{c}\sqrt{1+\alpha_n}}}}{\e^{-2a\frac{n\pi}{c}\sqrt{1-\alpha_n}}\left(1-\e^{-2\frac{n\pi}{c}\left(b\sqrt{1+\alpha_n}-a\sqrt{1-\alpha_n}\right)}\right)} \\
                &= \lim_{n\to\infty}\frac{\e^{-2a\frac{n\pi}{c}\sqrt{1-\alpha_n}} - \e^{-2b\frac{n\pi}{c}\sqrt{1+\alpha_n}}}{\alpha_n} \\
                &= \lim_{n\to\infty}\frac{\e^{-2a\frac{n\pi}{c}\sqrt{1-\alpha_n}}}{\alpha_n} \\
                &= \lim_{n\to\infty} \frac{n^2\e^{-2a\frac{n\pi}{c}\sqrt{1-\alpha_n}}}{n^2\alpha_n} = \lim_{n\to\infty} \frac{\e^{-n\left(\frac{2a\pi}{c}\sqrt{1-\alpha_n}-2\frac{\ln n}{n}\right)}}{n^2\alpha_n}
    \end{split}
\end{align}
where third and fourth equations hold because $b>a$. As the limit is finite and the numerator in the last term goes to zero, necessarily
\begin{equation}
    0=\lim_{n\to\infty}n^2 \alpha_n=\left(\frac{c}{\pi}\right)^2 \lim_{n\to\infty}\lambda_n.
\end{equation}
Thus, we have proven that $\lambda=0$ is the only accumulation point and that $\{0\}=\sigma_{\mathrm{ess}}(A_0)$. \end{proof}

\begin{corollary}\label{cor:convergence-rate}
     For $\epsilon_+=\epsilon_-$, the rate of convergence of eigenvalues of $A_0$ to $0$ is 
     \begin{equation}
         \min_{m\in\mathbb{Z}}|\lambda_{n,m}|=o\left(\e^{-\frac{n\pi}{c}\min\{a,b\}}\right).
     \end{equation}
\end{corollary}
\begin{proof}
    In addition to previous proposition, we can establish a rate of convergence. By using a similar trick as in proof of the proposition in the last equation, expand the limit expression (without loss of generality for $a<b$):
\begin{equation}
    \frac{1}{2}=\left(\frac{\pi}{c}\right)^2 \lim_{n\to\infty}\frac{\e^{-n\left(\frac{2a\pi}{c}\sqrt{1-\alpha_n}-2\frac{\ln n}{n}\right)}}{\lambda_n} \frac{\e^{a\frac{n\pi}{c}}}{\e^{a\frac{n\pi}{c}}}= \lim_{n\to\infty}\frac{\e^{-n\left(\frac{2a\pi}{c}\sqrt{1-\alpha_n}-\frac{a\pi}{c}-2\frac{\ln n}{n}\right)}}{\lambda_n \e^{a\frac{n\pi}{c}}}
\end{equation}
and by the same argument as before, we obtain the statement.
\end{proof}


\bibliographystyle{plain}
\bibliography{reference}

\end{document}